\documentclass[a4paper,fleqn]{cas-sc}

\usepackage[numbers]{natbib}

\def\tsc#1{\csdef{#1}{\textsc{\lowercase{#1}}\xspace}}
\tsc{WGM}
\tsc{QE}
\tsc{EP}
\tsc{PMS}
\tsc{BEC}
\tsc{DE}

\usepackage{amsthm}
\usepackage[ruled,vlined]{algorithm2e}
\usepackage{subfigure}
\usepackage{bm}

\newtheorem{theorem}{Theorem}[section]
\newtheorem{proposition}[theorem]{Proposition}
\newtheorem{lemma}[theorem]{Lemma}

\newtheorem{definition}[theorem]{Definition}

\newtheorem{remark}[theorem]{Remark}

\begin{document}
\let\WriteBookmarks\relax
\def\floatpagepagefraction{1}
\def\textpagefraction{.001}
\shorttitle{Revisiting Degree-Corrected Spectral Clustering: a Condition-Free Spectral Analysis and Extension}
\shortauthors{Wei Li et~al.}

\title [mode = title]{Revisiting Degree-Corrected Spectral Clustering: a Condition-Free Spectral Analysis and Extension}

\tnotetext[1]{This research was supported by the Interdisciplinary Frontier Research Project of PCL (2025QYB015), the National Key RD Program of China (2023YFA1010202), the Central Guidance on Local Science and Technology Development Fund of Fujian Province (2023L3003), the National Natural Science Foundation of China (11901094).}

\author[1]{Wei Li}
\author[2]{Xiaojian Li}
\author[3]{Meng Qin}
\cormark[1]
\ead{mengqin\_az@foxmail.com}
\author[4]{Chaorui Zhang}
\author[4]{Weixi Zhang}
\author[1]{Yiwen Zhong}
\author[2]{Jianfeng Hou}

\affiliation[1]{organization={School of Computer \& Information Science, Fujian Agriculture \& Forestry University},
city={Fuzhou},
state={Fujian},
country={China}}

\affiliation[2]{organization={School of Mathematics \& Statistics, Fuzhou University},
city={Fuzhou},
state={Fujian},
country={China}}

\affiliation[3]{organization={Department of Strategic \& Advanced Interdisciplinary Research, Pengcheng Laboratory (PCL)},
city={Shenzhen},
state={Guangdong}, 
country={China}}

\affiliation[4]{organization={Theory Lab, 2012 Labs, Huawei},
state={Hong Kong SAR}, 
country={China}}

\cortext[cor1]{Corresponding author}

\begin{abstract}
Spectral clustering is a representative graph clustering technique with strong interpretability and theoretical guarantees. Degree-corrected spectral clustering (DCSC) has emerged as the state-of-the-art for this technique. While prior studies have provided impressive theoretical insights for DCSC, their analyses typically depend on specific probabilistic frameworks (e.g., stochastic block models) and conditions. In this study, we explore an alternative condition-free analysis for the clustering quality of DCSC from a pure spectral view, without any random graph models. It gives bounds for the number of mis-clustered nodes w.r.t. the optimal partition of conductance minimization while involving quantities that indicate impacts of (\romannumeral1) degree heterogeneity and (\romannumeral2) weakness of clustering structures to the clustering quality. Inspired by graph neural networks (GNNs) and their over-smoothing effect, we propose ASCENT (Adaptive Spectral ClustEring with Node-wise correcTion), a simple yet effective extension of DCSC. Different from most DCSC methods with a constant degree correction, ASCENT follows a node-wise correction scheme. It can assign different corrections for nodes via a GNN mean aggregator. We demonstrate that (\romannumeral1) ASCENT reduces to conventional DCSC methods when encountering over-smoothing; (\romannumeral2) some early stages before over-smoothing can potentially result in better clustering quality.
\end{abstract}


\maketitle

\section{Introduction}\label{Sec:Intro}
Graph clustering (a.k.a. disjoint community detection) is a classic inference task that partitions nodes of a graph into densely connected groups (i.e., clusters or communities). Since the extracted clusters have been validated to correspond to some substructures of real-world systems (e.g., functional groups in protein interactions \cite{Berahmand}), many network applications are formulated as graph clustering, including protein complex detection \cite{qin2010spectral}, cellular network decomposition \cite{dai2017optimal}, and Internet traffic profiling \cite{qin2019towards}.

Spectral clustering (SC) is one of the representative techniques for this task.
As summarized in Table~\ref{Tab:Alg-Sum}, a typical SC algorithm includes the (\uppercase\expandafter{\romannumeral1}) eigen-decomposition (ED) on graph Laplacian, (\uppercase\expandafter{\romannumeral2}) arrangement of spectral embeddings, (\uppercase\expandafter{\romannumeral3}) normalization of embeddings, and (\uppercase\expandafter{\romannumeral4}) $K$Means clustering.
In Table~\ref{Tab:Alg-Sum}, ${\bf{A}}$ and ${\bf{D}}$ are the adjacency matrix and degree diagonal matrix of a graph; $K$ is a pre-set number of clusters; $\lambda_r$ denotes the $r$-th largest eigenvalue of graph Laplacian ${\bf{L}}$ (e.g., ${\bf{L}} := {\bf{D}}^{-1/2} {\bf{A}} {\bf{D}}^{-1/2}$ and ${\bf{L}} := {\bf{A}}$ for \textit{NJW} \cite{NJW} and \textit{SCORE} \cite{jin2015fast}) with ${\bf{u}}_r \in \mathbb{R}^{N}$ as the corresponding eigenvector.
Different SC algorithms usually differ in terms of the four steps. For instance, \textit{NJW}, \textit{SCORE}, and \textit{RSC} \cite{RSC} only consider eigenvectors $({\bf{u}}_1, \cdots, {\bf{u}}_K)$ w.r.t. the leading $K$ eigenvalues. Step (\uppercase\expandafter{\romannumeral2}) of \textit{SCORE+} \cite{jin2021improvements} and \textit{ISC} \cite{ISC} involves $({\bf{u}}_1, \cdots, {\bf{u}}_K, {\bf{u}}_{K+1})$, which are further reweighted by corresponding $(K+1)$ eigenvalues $(\lambda_1, \cdots, \lambda_K, \lambda_{K+1})$.
\textit{NJW}, \textit{RSC}, and \textit{ISC} adopt the row-wise $l_2$-normalization in step (\uppercase\expandafter{\romannumeral3}). \textit{SCORE} and \textit{SCORE+} use (reweighted) ${\bf{u}}_1$ to conduct column-wise normalization.

Degree-corrected spectral clustering (DCSC), which is also defined as regularized spectral clustering in some literature \cite{RSC,zhang2018understanding},
has emerged as a state-of-the-art class of SC methods, due to its effectiveness in handling the high degree heterogeneity. Typical DCSC approaches usually incorporate an additional degree correction term $\tau$ in their graph Laplacian for ED (e.g., \textit{RSC}, \textit{SCORE+}, and \textit{ISC} with different settings of $\tau$ in Table~\ref{Tab:Alg-Sum}).

\begin{table}[]\scriptsize
\caption{Summary of representative SC algorithms, where ${\bf{D}}_\tau := {\bf{D}} + \tau {\bf{I}}_N$; $\tau$ is the degree correction term in DCSC, with $\tau = \bar d$, $\delta d_{\max}$, and $\delta (d_{\min} + d_{\max})/2$ for \textit{RSC}, \textit{SCORE+}, and \textit{ISC} (e.g., $\delta = 0.1$); $\bar d$, $d_{\min}$, and $d_{\max}$ are the average, minimum, and maximum degrees.}\label{Tab:Alg-Sum}
\centering
\begin{tabular}{l|l|l|l|l}
\hline
 & Step (\uppercase\expandafter{\romannumeral1}) & Step (\uppercase\expandafter{\romannumeral2}) & Step (\uppercase\expandafter{\romannumeral3}) & Step (\uppercase\expandafter{\romannumeral4}) \\ \hline
\textbf{\textit{NJW}} & ED on ${\bf{D}}^{-1/2}{\bf{A}}{\bf{D}}^{-1/2}$ & ${\bf{F}} := [{\bf{u}}_1, \cdots, {\bf{u}}_K]$ & $i \in [1, N]$, ${{\bf{F}}_{i,:}} \leftarrow {{\bf{F}}_{i,:}}/|{{\bf{F}}_{i,:}}{|_2}$ & \multirow{5}{*}{\begin{tabular}[c]{@{}l@{}}$K$Means on \\${\bf{F}}$'s rows\end{tabular}} \\ \cline{1-4}
\textbf{\textit{SCORE}} & ED on ${\bf{A}}$ & ${\bf{F}} := [{\bf{u}}_2, \cdots, {\bf{u}}_K ]$ & $r \in [1, K-1]$, ${{\bf{F}}_{:,r}} \leftarrow {{\bf{F}}_{:,r}}/{{\bf{u}}_1}$ &  \\ \cline{1-4}
\textbf{\textit{RSC}} & \multirow{3}{*}{ED on ${\bf{D}}_{\tau}^{-1/2}{\bf{A}}{\bf{D}}_{\tau}^{-1/2}$} & ${\bf{F}} := [{\bf{u}}_1, \cdots, {\bf{u}}_K]$ & $i \in [1, N]$, ${{\bf{F}}_{i,:}} \leftarrow {{\bf{F}}_{i,:}}/|{{\bf{F}}_{i,:}}{|_2}$ &  \\ \cline{1-1} \cline{3-4}
\textbf{\textit{SCORE+}} &  & ${\bf{F}}: = [{\lambda _2}{{\bf{u}}_2}, \cdots ,{\lambda _{K + 1}}{{\bf{u}}_{K + 1}}]$ & $r \in [1, K]$, ${{\bf{F}}_{:,r}} \leftarrow {{\bf{F}}_{:,r}} / (\lambda_1 {{\bf{u}}_1})$ &  \\ \cline{1-1} \cline{3-4}
\textbf{\textit{ISC}} &  & ${\bf{F}}: = [{\lambda _1}{{\bf{u}}_1}, \cdots ,{\lambda _{K + 1}}{{\bf{u}}_{K + 1}}]$ & $i \in [1, N]$, ${{\bf{F}}_{i,:}} \leftarrow {{\bf{F}}_{i,:}}/|{{\bf{F}}_{i,:}}{|_2}$ &  \\ \hline
\end{tabular}
\end{table}

\begin{table}[]\scriptsize
\caption{Summary of representative theoretical analysis on clustering quality of DCSC and the comparison to ours.}\label{Tab:Ana-Sum}
\centering
\begin{tabular}{l|l|l|l}
\hline
 & \textbf{\begin{tabular}[c]{@{}l@{}}Random \\ Graph  Models\end{tabular}} & \textbf{\begin{tabular}[c]{@{}l@{}}Condition\\-Free\end{tabular}} & \textbf{Quality Bounds} \\ \hline
Coja-Oghlan (2010) \cite{coja2010graph} & Any Models & No & Any random graph's recovery error \\ \hline
Chaudhuri et al. (2012) \cite{Chaudhuri} & EPP Model & No & EPP's optimal seperation \\ \hline
Qin and Rohe (2013) \cite{RSC} & DCSBM & No & \multicolumn{1}{l}{\multirow{3}{*}{Mis-clustered rate w.r.t. SBM's ground-truth}} \\ \cline{1-3}
Amini et al. (2013) \cite{amini2013pseudo} & DCSBM & No & \multicolumn{1}{l}{} \\ \cline{1-3}
Cucuringu et al. (2021) \cite{cucuringu2021regularized} & Signed SBM & No & \multicolumn{1}{l}{} \\ \hline
Joseph and Yu (2016) \cite{joseph2016impact} & DCSBM & No & \multicolumn{1}{l}{\multirow{3}{*}{Error rate w.r.t. SBM's ground-truth}} \\ \cline{1-3}
Qing and Wang (2020) \cite{ISC} & DCSBM & No & \multicolumn{1}{l}{} \\ \cline{1-3}
Jin et al. (2021) \cite{jin2021improvements} & DCSBM & No & \multicolumn{1}{l}{} \\ \hline
Dall’Amico et al. (2019) \cite{dall2019revisiting} & DCSBM & No & Overlap w.r.t. SBM's bisection ground-truth \\ \hline
Cohen-Addad et al. (2022) \cite{cohen2022community} & DHSBM & No & Recovery probability w.r.t. SBM's   ground-truth \\ \hline
\textbf{Ours} & \textbf{w/o any models} & \textbf{Yes} & \#mis-clustered nodes w.r.t. optimal graph-cut \\ \hline
\end{tabular}
\end{table}

\textbf{Related Analysis for the Clustering Quality of DCSC}.
In the past few decades, a series of SC methods have been proposed. Ding et al.~\cite{ding2024survey} provided an overview of related research.
Table~\ref{Tab:Ana-Sum} summarizes some representative theoretical results about the clustering quality of DCSC.

For instance, Chaudhuri et al.~\cite{Chaudhuri} proposed a DCSC method for graphs drawn from an extended planted partition (EPP) model \cite{condon2001algorithms} and examined its quality guarantees w.r.t. the optimal separation of EPP.
Qin and Rohe~\cite{RSC} analyzed the quality bound of \textit{RSC} in terms of mis-clustered rate using the degree-corrected stochastic blockmodel (DCSBM) \cite{karrer2011stochastic} and provided guidance on the choice of $\tau$.
Also based on DCSBM, Qing et al.~\cite{ISC} and Jin et al~\cite{jin2021improvements} respectively proposed \textit{ISC} and \textit{SCORE+}, which have theoretical guarantees about the Hamming error rate w.r.t. ground-truth.
Cohen-Addad et al.~\cite{cohen2022community} introduced the degree-heterogeneous stochastic blockmodel (DHSBM) and proved that there exists a linear-time spectral algorithm with bounded recovery probability for DHSBM.

In summary, most related theoretical studies rely on assumptions of random graph models (e.g., DCSBM). They usually fit the adjacency matrix or graph Laplacian using a certain random graph model (e.g., $\mathcal{A}: = {\bf{\Theta ZB}}{{\bf{Z}}^T}{\bf{\Theta }}$ \cite{RSC,ISC} with $\{ {\bf{\Theta}}, {\bf{Z}}, {\bf{B}} \}$ defined in DCSBM) and give bounds related to such a model (e.g., mis-clustered rate or error w.r.t. the ground-truth of DCSBM in Table~\ref{Tab:Ana-Sum}).

Moreover, most of them presuppose some additional conditions or assumptions (e.g., $[K\ln (4N/\epsilon )/(\bar d_{\min}  + \tau )]^{1/2} \le {\lambda _K}/(8\sqrt 3 )$ and ${\bar d_{\min }} + \tau  > 3\ln (4N/\varepsilon )$ for \textit{RSC} in \cite{RSC}; ${[\ln (4N/\varepsilon )/({{\bar d}_{\min }} + \tau )]^{1/2}} \le {\lambda _K}/8\sqrt 3$ and ${{\bar d}_{\min }} + \tau  > 3\ln (4N/\varepsilon )$ for \textit{ISC} in \cite{ISC}), which may fail in some extreme cases (e.g., for a large number of nodes $N$) and thus result in trivial theoretical results.

\textbf{Present Analysis \& Extension}.
SC is a typical approximated algorithm for the NP-hard combinatorial optimization of graph-cut minimization \cite{von2007tutorial}. From this perspective, some prior work \cite{peng,Mizutani,macgregor2022tighter} analyzed vanilla SC (e.g., \textit{NJW}) using the spectral graph theory.
Motivated by these studies, we consider \textbf{an alternative condition-free analysis for the clustering quality of DCSC from a pure spectral view instead of using any random graph models}.

Different from existing analyses with bounds related to random graph models (e.g., DCSBM in Table~\ref{Tab:Ana-Sum}), we provide bounds for the number of mis-clustered nodes w.r.t. the optimal solution to average conductance minimization.
In contrast to early spectral-based studies on vanilla SC \cite{peng,Mizutani}, our analysis also involves quantities about (\romannumeral1) degree heterogeneity and (\romannumeral2) weakness of clustering structures, which can reveal impacts of these two aspects to the clustering quality of DCSC.

Inspired by graph neural networks (GNNs), we propose ASCENT (\underline{A}daptive \underline{S}pectral \underline{C}lust\underline{E}ring with \underline{N}ode-wise correc\underline{T}ion), a simple yet effective extension of DCSC.
Instead of using a constant correction $\tau$ for all nodes (e.g., \textit{RSC}, \textit{SCORE+}, and \textit{ISC} in Table~\ref{Tab:Alg-Sum}), ASCENT adopts a node-wise correction scheme, where nodes $\{ v_i \}$ can be assigned with different corrections $\{ \tau_i \}$. This scheme iteratively updates $\{ \tau_i \}$ via a GNN mean aggregator, where nodes $\{ v_i \}$ with more common neighbors (e.g., in the same cluster) are more likely to have closer $\{ \tau_i \}$. Consistent with the over-smoothing effect of GNNs, $\{ \tau_i \}$ finally converge to a constant. In this case, ASCENT reduces to conventional DCSC methods. Our experiments demonstrate that some early stages of this updating procedure (i.e., before over-smoothing) can potentially result in better clustering quality.

We also extend our analysis to a unified framework that provides bounds for a series of SC methods. Based on this framework, we try to answer the following three questions.
\begin{itemize}
    \item \textbf{Q1}: Can DCSC achieve a tighter bound than vanilla SC?
    \item \textbf{Q2}: Can \textit{ISC} achieve a tighter bound than \textit{RSC}?
    \item \textbf{Q3}: Can ASCENT achieve a tighter bound than \textit{ISC}?
\end{itemize}
Answering \textbf{Q1}-\textbf{Q3} may help interpret why a DCSC method outperforms others. Whereas, most related studies tend to focus on a single algorithm and may not be able to answer these questions explored in this paper.

The remainder of this paper is organized as follows.
Section~\ref{Sec:Rel} reviews related work. Formal problem statements and preliminaries of this study are given in Section~\ref{Sec:Prob}. Section~\ref{Sec:Ana} presents our condition-free spectral analysis for the clustering quality of DCSC, while Section~\ref{Sec:Meth} elaborates on the extended ASCENT algorithm. Empirical experiments are further described in Section~\ref{Sec:Exp}. Finally, Section~\ref{Sec:Con} concludes this paper and discusses possible future research directions to tackle limitations of this study.

\section{Related Work}\label{Sec:Rel}

\begin{table}[]\scriptsize
\caption{Summary of other theoretical studies in addition to clustering quality of DCSC and the comparison to ours.}\label{Tab:Ana-Sum-Other}
\centering
\begin{tabular}{l|l|l|p{6.5cm}}
\hline
 & \textbf{\begin{tabular}[c]{@{}l@{}}Random\\ Graph  Models\end{tabular}} & \textbf{\begin{tabular}[c]{@{}l@{}}Condition-Free\end{tabular}} & \textbf{Main (Theoretical) Results} \\ \hline
Lee et al. (2014) \cite{lee2014multiway} & w/o any models & Yes & New \textbf{spectrum} \& \textbf{topological properties} of graph Laplacian \\ \hline
Krzakala et al. (2013) \cite{krzakala2013spectral} & DCSBM & {No theoretical analysis} & \textbf{Spectrum properties} \& \textbf{detectability} of non-backtracking  matrix in DCSBM \\ \hline
Gulikers et al. (2017) \cite{gulikers2017non} & DCSBM & No & \textbf{Spectrum properties} \& \textbf{detectability} of non-backtracking matrix in DCSBM \\ \hline
Dall’Amico et al. (2020) \cite{dall2020optimal} & DCSBM & No & \textbf{Relation} b/w Bethe-Hessian \& graph Laplacian, and \textbf{detectability} in DCSBM \\ \hline
Dall’Amico et al. (2021) \cite{dall2021unified} & DCSBM & No & \textbf{Spectrum properties} \& \textbf{detectability} of Bethe-Hessian matrix in DCSBM \\ \hline
Dall’Amico et al. (2021) \cite{dall2021nishimori} & ER Graph & No & \textbf{Relation} b/w Nishimori temperature \& Bethe free energy on ER graphs \\ \hline
\textbf{Ours} & \textbf{w/o any models} & \textbf{Yes} & \textbf{Clustering quality} of a set of DCSC algorithms \\ \hline
\end{tabular}
\end{table}

\subsection{Theoretical Analyses of (Degree-Corrected) Spectral Clustering}

In addition to the theoretical analyses on clustering quality shown in Table~\ref{Tab:Ana-Sum}, there are studies focusing on other aspects about (DC)SC as summarized in Table~\ref{Tab:Ana-Sum-Other}.
For instance, Lee et al. \cite{lee2014multiway} provided a theoretical justification of using leading $K$ eigenvectors $[{\bf{u}}_1, \cdots, {\bf{u}}_K]$ for SC and investigated associated properties about graph topology.
Based on the detectability analyses on DCSBM \cite{decelle2011asymptotic,massoulie2014community}, several studies try to reveal inherent spectrum properties (e.g., distribution of leading eigenvalues) of the non-backtracking matrix \cite{krzakala2013spectral,gulikers2017non} and Bethe-Hessian matrix \cite{dall2020optimal,dall2021unified,dall2021nishimori}, which were further used to design new (DC)SC algorithms.

While we acknowledge prior studies summarized in Table 3, which explore spectral properties of specific matrices or detectability under certain random graph models, we note that their focus differs from the main objective of this work (i.e., the analysis of quality bounds for DCSC). Consequently, a direct quantitative comparison between these prior results and our theoretical bounds may not be readily feasible.
Furthermore, most of them still rely on specific conditions and random graph models (e.g., DCSBM).
In contrast, our analysis offers a distinct perspective by providing a condition-free bound derived purely from spectral theory, without recourse to any underlying random graph model.

\subsection{Applications of Deep Graph Clustering}

In recent years, several deep graph clustering (DGC) methods have been proposed as reviewed by Liu et al.~\cite{yue2022survey} and Su et al.~\cite{su2022comprehensive}.
\textit{GraphEncoder} \cite{tian2014learning} and \textit{DNR} \cite{yang2016modularity} are early studies that learn low-dimensional community-preserving representations (or embeddings) by reconstructing topology-related features (e.g., normalized adjacency matrices and modularity matrices) via a deep auto-encoder.
\textit{THESAURUS} \cite{deng2025thesaurus} enhances the ability of learned embeddings to capture implicit graph clustering structures by using (\romannumeral1) context-aware semantic prototypes, (\romannumeral2) a cross-view assignment prediction objective, and (\romannumeral3) Gromov-Wasserstein Optimal Transport.

For clustering on large-scale graphs, Devvrit et al.~\cite{devvrit2022s3gc} proposed \textit{S3GC}, which learns scalable community-preserving embeddings via GNN-based contrastive learning. \textit{MAGI} \cite{liu2024revisiting} extends \textit{S3GC} to incorporate the modularity maximization objective \cite{newman2006modularity} from the view of contrastive learning.
A downstream clustering algorithm (e.g., $K$Means) is usually applied to the learned embeddings to derive a feasible clustering result for these approaches.
\textit{DGCluster} \cite{bhowmick2024dgcluster} also uses the modularity maximization objective to optimize GNN, which outputs community-preserving embeddings, but derives final clustering results using BIRCH \cite{zhang1996birch}.

\textit{ClusterNet} \cite{wilder2019end}, \textit{MinCutPool} \cite{bianchi2020spectral}, and \textit{DMoN} \cite{tsitsulin2023graph} adopt a deep end-to-end structure, which contains a GNN and an output module (e.g., a multi-layer perceptron for deriving clustering results), to fit classic graph clustering objectives (e.g., graph-cut minimization and modularity maximization).
\textit{SDCN} \cite{bo2020structural} combines the deep auto-encoder with GNN and uses a dual self-supervised mechanism to unify these two deep architectures.

Most of the aforementioned methods, especially those based on GNNs, were originally designed for attributed graphs and may not fully consider the complicated correlations between graph topology and attributes (see Section~\ref{Sec:Prob} for further discussions).
Our empirical experiments (see Section~\ref{Sec:Exp}) demonstrated that when attributes are unavailable, these state-of-the-art DGC methods may not outperform SC approaches.
Different from SC, most related studies about DGC may also lack interpretability and theoretical guarantees.

\section{Problem Statements \& Preliminaries}\label{Sec:Prob}
In general, an undirected and unweighted simple graph can be represented as a 2-tuple $G := (V,E)$, where $V := \{ v_1, \cdots, v_N\}$ and $E := \{ (v_i, v_j)| v_i, v_j \in V\}$ are sets of nodes and edges. One can use an adjacency matrix ${\bf{A}} \in \{ 0, 1\} ^ {N \times N}$ to describe the topology of $G$, where ${\bf{A}}_{ij} = {\bf{A}}_{ji} = 1$ if $(v_i, v_j) \in E$ and ${\bf{A}}_{ij} = {\bf{A}}_{ji} = 0$ otherwise. Let ${\bf{D}} := {\mathop{\rm diag}\nolimits} (d_1, d_2, \cdots, d_N)$ be the degree diagonal matrix of $G$, with $d_i := \sum\nolimits_j {{{\bf{A}}_{ij}}}$ as the degree of node $v_i$.

Given a graph $G$ and a number of clusters $K$, \textbf{graph clustering} (a.k.a. \textbf{disjoint community detection}) aims to partition $V$ into $K$ disjoint subsets $(C_1, \cdots, C_K)$, defined as clusters or communities, such that (\romannumeral1) within each cluster the edge connections between nodes are dense but (\romannumeral2) between clusters the connections are relatively loose.

We follow the classic problem statement of SC, where graph topology is the only available information source. Different from most DGC methods \cite{bo2020structural,bianchi2020spectral,tsitsulin2023graph,bhowmick2024dgcluster}, our analysis does not consider graph attributes, due to the complicated correlations between topology and attributes validated by prior studies \cite{qin2018adaptive,wang2020gcn,qin2021dual}. Concretely, the simple integration of attributes may bring inconsistent features or noise that lead to quality decline compared with only considering topology, although attributes may sometimes provide complementary information for better clustering quality.

For a subset $C\subseteq V$, let $E(C, V \backslash C) := \{(v_i, v_j)\in E: v_i \in C, v_j \in V \backslash C \}$ be the set of edges across $C$ and $V \backslash C$. Let $\mu(C) := \sum_{v_i \in C} d_i$ be the \textbf{volume} of $C$.
The \textbf{conductance} of $C$ is defined as $\phi (C): = |E(C,V\backslash C)|/\mu (C)$.

\begin{definition}[Average Conductance Minimization, \cite{Mizutani}]\label{Def:Cond-Min}
Let $U$ be the collection of all possible $K$-way partitions of $V$ in $G$. The \textbf{average conductance minimization} objective\footnote{It is equivalent to the \textbf{normalized cut minimization} objective that $\min[\phi (C_1) + \cdots + \phi (C_K)]$.} is defined as
\begin{equation}\label{Eq:Cond}
    {\bar \phi _K}(G): = {\min}_{({C_1}, \cdots ,{C_K}) \in U}~ [\phi ({C_1}) +  \cdots  + \phi ({C_K})]/K.
\end{equation}
It finds a partition $(C_1, \cdots, C_K)$ that achieves the \textbf{minimal average conductance} $\bar{\phi}_K(G)$. We further define that a partition $(C_1, \cdots, C_K)$ is $\bar{\phi}_K(G)$-\textbf{optimal} if its average conductance achieves $\bar{\phi}_K(G)$.
\end{definition}

For the ED on graph Laplacian (i.e., step (\uppercase\expandafter{\romannumeral1}) of Table~\ref{Tab:Alg-Sum}), let $\lambda_r$ and ${\bf{u}}_r \in \mathbb{R}^{N}$ denote the $r$-th largest eigenvalue and corresponding eigenvector.
Considering the normalized graph Laplacian ${\bf{D}}^{-1/2} {\bf{A}} {\bf{D}}^{-1/2}$, we have $1 = \lambda_1 \ge \cdots \ge \lambda_N \geq -1$\footnote{Some literature \cite{von2007tutorial,qin2023towards,gao2023raftgp} defines the normalized graph Laplacian as ${\bf{I}}_N - {\bf{D}}^{-1/2}{\bf{A}}{\bf{D}}^{-1/2}$, which equivalently has the eigenvalues of $0 = 1 - {\lambda _1} \le  \cdots  \le 1 - {\lambda _N} \le 2$.} and ${\bf{u}}_r^T{{\bf{u}}_t} = 0$ ($\forall r \ne t$). Moreover, we have $1 > \lambda_1 \ge \lambda_2 \ge \cdots \ge \lambda_N$ for the regularized graph Laplacian ${\bf{D}}_{\tau}^{-1/2} {\bf{A}} {\bf{D}}_{\tau}^{-1/2}$. In step (\uppercase\expandafter{\romannumeral2}) of Table~\ref{Tab:Alg-Sum}, we arrange the (reweighted) eigenvectors as a matrix ${\bf{F}} \in \mathbb{R}^{N \times K}$ (or $\mathbb{R}^{N \times (K+1)}$) via the column-wise concatenation. We define the $i$-th row ${\bf{F}}_{i,:}$ of ${\bf{F}}$ as the \textbf{spectral embedding} of node $v_i$. Most SC algorithms apply normalization to ${\bf{F}}$ (i.e., step (\uppercase\expandafter{\romannumeral3}) in Table~\ref{Tab:Alg-Sum}). We denote the \textbf{normalized spectral embeddings} as ${\bf{\tilde F}}$.

\begin{definition}[Clustering Cost, \cite{peng}]\label{Def:KMeans}
Given vectors $({\bf{w}}_1, \cdots, {\bf{w}}_K)$,  the \textbf{distance} between a partition $(C_1, \cdots, C_K)$ and $({\bf{w}}_1, \cdots, {\bf{w}}_K)$ is defined as
\begin{equation}\label{Eq:Dist}
    g({C_1} \cdots ,{C_K};{{\bf{w}}_1}, \cdots ,{{\bf{w}}_K}): = \sum\nolimits_{r = 1}^K {\sum\nolimits_{{v_i} \in {C_r}} {{d_i}||{{{\bf{\tilde F}}}_{i,:}} - {{\bf{w}}_r}||_2^2} },
\end{equation}
where we map each node $v_i$ to $d_i$ identical points in the embedding space.
The \textbf{clustering cost} of a partition $(C_1, \cdots, C_K)$ is then defined as
\begin{equation}\label{Eq:COST}
    {\mathop{\rm COST}\nolimits} ({C_1, \cdots, C_K}): = {\min }_{({{\bf{c}}_1, \cdots, {\bf{c}}_K})}~g({C_1, \cdots, C_K} ;{{\bf{c}}_1, \cdots, {\bf{c}}_K}),
\end{equation}
which finds a set of centers $({\bf{c}}_1, \cdots, {\bf{c}}_K)$ with the minimum distance to $(C_1, \cdots, C_K)$. Based on ${\mathop{\rm COST}\nolimits} ({C_1}, \cdots ,{C_K})$, we define the \textbf{optimal clustering cost} as
\begin{equation}\label{Eq:OPT}
    {\mathop{\rm OPT}\nolimits} : = {\min }_{({C_1}, \cdots ,{C_K}) \in U}~{\mathop{\rm COST}\nolimits} ({C_1, \cdots, C_K}).
\end{equation}
\end{definition}
As claimed in \cite{peng}, this definition allows us to bound the overlap between (\romannumeral1) feasible clustering results and (\romannumeral2) optimal ones, which is used in our analysis. By assuming that for each node $v_i \in V$, all the $d_i$ copies of ${\bf{\tilde F}}_{i,:}$ are contained in one of $\{ C_1, \cdots, C_K \}$, (\ref{Eq:Dist}) reduces to the standard distance in $K$Means.

\section{Proposed Analysis on Clustering Quality of DCSC: A Spectral View}\label{Sec:Ana}
Inspired by prior spectral-based studies on vanilla SC \cite{peng,Mizutani,macgregor2022tighter}, we propose a condition-free analysis for the clustering quality of DCSC from a pure spectral view instead of using random graph models.

We adopt \textit{ISC} (see Table~\ref{Tab:Alg-Sum}) as an example for analysis because it has a more generic format involving the reweighted $(K+1)$ leading eigenvectors $[ {\lambda_1} {{\bf{u}}_1}, \cdots, {\lambda_{K+1}} {{\bf{u}}_{K+1}} ]$. Whereas, other DCSC methods have simpler formats (e.g., only $[{\bf{u}}_1, \cdots, {\bf{u}}_K]$ for \textit{RSC}).
We then reduce this analysis to other algorithms (e.g., \textit{NJW}, \textit{SCORE+}, and \textit{RSC}), which form a unified framework with theoretical bounds for a series of (DC)SC approaches.

\subsection{Basic Theoretical Analysis on \textit{ISC}}

In contrast to prior work \cite{peng,Mizutani,macgregor2022tighter} on vanilla SC, our analysis also aims to reveal impacts of (\romannumeral1) degree heterogeneity and (\romannumeral2) weakness of clustering structures to DCSC. We first introduce a quantity measuring both aspects:
\begin{equation}\label{Eq:Psi}
    \Psi_{\rm{ISC}} : = m_K^{ - 1}[1 - {\tilde d} (1 - {{\bar \phi} _K}(G))] = {(1 - {\lambda _{K + 2}})^{ - 1}}[1 - {d_{\min }}(1 - {{\bar \phi }_K}(G))/({d_{\max }} + \tau )],
\end{equation}
with $m_K := 1 - \lambda_{K+2}$ and ${\tilde d}: = {d_{\min }}/({d_{\max }} + \tau )$.
In (\ref{Eq:Psi}), ${\tilde d}$ measures the degree heterogeneity, where a small ${\tilde d}$ (i.e., a large difference between $d_{\min}$ and $d_{\max}$) indicates high degree heterogeneity. Since ${\bar \phi _K}(G) \le 1$, \textit{higher degree heterogeneity (i.e., a smaller ${\tilde d}$) will lead to a larger $\Psi_{\rm{ISC}}$}.

As validated in \cite{jin2021improvements}, when clustering structures of a graph (with $K$ clusters) are weak, ${{\tilde m}_K}: = 1 - {\lambda _{K + 1}}/{\lambda _K}$ is small, consistent with a small $|\lambda_K - \lambda_{K+1}|$ by the \textbf{eigen-gap property} \cite{von2007tutorial}. Since $1 > {\lambda _K} \ge {\lambda _{K + 1}} \ge {\lambda _{K + 2}}$, ${m_K} = 1 - {\lambda _{K + 2}} \ge 1 - {\lambda _{K + 1}} \ge 1 - {\lambda _{K + 1}}/{\lambda _K} = {\tilde m_K}$ and thus
\begin{equation*}
    {\Psi _{{\rm{ISC}}}} = {m_K^{-1}}[1 - {\tilde d}(1 - {\bar \phi _K}(G))] \le {\tilde m}_K^{-1}[1 - {\tilde d}(1 - {\bar \phi _K}(G))].
\end{equation*}
\textit{Weaker clustering structures (i.e., a smaller ${\tilde m}_K$) indicate a larger upper bound of $\Psi_{\rm{ISC}}$}.

\begin{theorem}\label{Th:Struc}
    Let $({\hat S}_1, \cdots, {\hat S}_K)$ be a $\bar{\phi}_K(G)$-\textbf{optimal} partition, with the partition membership encoded by ${\bf{G}} \in \mathbb{R}^{N \times K}$. ${\bf{G}}_{ir} = [d_i / \mu({\hat S}_r)]^{1/2}$ if $v_i \in {\hat S}_r$ and ${\bf{G}}_{ir} = 0$ otherwise.
    ${\bf{ F}} := [\lambda_1 {\bf{u}}_1, \cdots, \lambda_{K+1} {\bf{u}}_{K+1}] \in \mathbb{R}^{N \times (K+1)}$ is the \textbf{spectral embedding} of ISC (i.e., step (\uppercase\expandafter{\romannumeral2}) of Table~\ref{Tab:Alg-Sum}). There exists an orthogonal matrix ${\bf{O}}:= [{\bf{o}}_1, \cdots, {\bf{o}}_K] \in \mathbb{R}^{(K+1) \times K}$ s.t.
    \begin{equation}\label{eq0}
    ||{\bf{FO}} - {\bf{G}}|{|_F} \le \left\{ {\begin{array}{*{20}{l}}
    {(1 + {\lambda _1})\sqrt {K{\Psi _{{\rm{ISC}}}}},~K{\Psi _{{\rm{ISC}}}} \le 1}\\
    {(1 + {\lambda _1})K{\Psi _{{\rm{ISC}}}},~K{\Psi _{{\rm{ISC}}}} > 1}
    \end{array}} \right..
    \end{equation}
\end{theorem}
\textbf{Theorem~\ref{Th:Struc}} can be proved by reformulating ${\bf{G}}$ via the linear combination of orthogonal eigenvectors $\{ {\bf{u}}_i \}$. Please refer to Appendix~\ref{App:Th-Struc} for the proof.
In the rest of this paper, we consider the case of $K\Psi_{\rm{ISC}} \le 1$ for simplicity. One can easily extend our results to that of $K\Psi_{\rm{ISC}} > 1$ by simply replacing $K\Psi_{\rm{ISC}}$ with $K^2\Psi_{\rm{ISC}}^2$.
The first term in (\ref{eq0}) can be further rewritten as
\begin{equation}\label{Eq:Struc}
    || {\bf{F}} {\bf{O}} - {\bf{G}} ||_F^2 = || {\bf{F}} - {\bf{G}} {\bf{O}}^T ||_F^2= || {\bf{F}}^T - {\bf{O}} {\bf{ G}}^T ||_F^2 =\sum\nolimits_{r} {\sum\nolimits_{{v_i} \in {{\hat S}_r}} {||{{\bf{F}}_{i,:}} - \sqrt {{d_i}/\mu ({{\hat S}_r})} {{\bf{o}}_r}||_2^2} }.
\end{equation}

By using a strategy similar to the proof of Lemma 2 in \cite{Mizutani}, we can derive the following \textbf{Lemma~\ref{le1}} with an upper bound for \textbf{clustering cost}. Appendix~\ref{App:Lemma1} gives the proof, which connects (\ref{Eq:Struc}) with (\ref{Eq:Dist}).
\begin{lemma}\label{le1}
    Let $({\hat S}_1, \cdots, {\hat S}_K)$ be a $\bar{\phi}_K(G)$-\textbf{optimal} partition and ${\bf{\tilde F}}$ be the \textbf{normalized spectral embedding} of \textit{ISC}. $\{ {\bf{o}}_r\}$ are with the same definitions as those in \textbf{Theorem~\ref{Th:Struc}}. The following equation and inequality hold:
    \begin{itemize}
        \item (\romannumeral1) $\left\| {{{\bf{o}}_r} - {{\bf{o}}_t}} \right\|_2^2=2 {\rm{~}}$, $\forall r,t \in \{ 1, 2, \cdots, K\}$ and $r \ne t$;
        \item (\romannumeral2) $g({{\hat S}_1, \cdots, {\hat S}_K};{{\bf{o}}_1, \cdots, {\bf{o}}_K}) \le 4{(1 + {\lambda _1})^2}{\mu _{\max }}K\Psi_{\rm{ISC}}$, with ${\mu _{\max }}: = \mathop {\max }_{{\hat S_r}} \{ \mu ({\hat S_r}) \}$.
    \end{itemize}
\end{lemma}

Obviously, we have ${\mathop{\rm OPT}\nolimits} \leq {\mathop{\rm COST}\nolimits} (C_1, \cdots, C_K) \leq g({\hat S}_1, \cdots, {\hat S}_K; {\bf{o}}_1, \cdots, {\bf{o}}_K)$.
Let $\alpha$ be the approximation ratio of $K$Means, i.e., $\mbox{COST}(C_1, \cdots, C_K) \le \alpha {\mathop{\rm OPT}\nolimits}$.
One can derive the following \textbf{Theorem~\ref{th2-1}} based on \textbf{Lemma~\ref{le1}}.
\begin{theorem}\label{th2-1}
    Let $(C_1, \cdots, C_K )$ be a feasible clustering result given by \textit{ISC}. When the $K$Means clustering algorithm has an approximation ratio of $\alpha$, we have
    \begin{equation}
        {\mathop{\rm COST}\nolimits} (C_1, \cdots, C_K) \le   4 (1 + \lambda_1)^2 \alpha \mu_{\max}K\Psi_{\rm{ISC}}.
    \end{equation}
\end{theorem}

\begin{lemma}\label{le2}
    Let $\pi (r): = \arg {\min _{t \in [K]}}||{{\bf{o}}_t} - {{\bf{c}}_r}||_2$ be the index of an orthogonal base in ${\bf{O}}$ (derived by \textbf{Theorem~\ref{Th:Struc}}) that has the closest distance to the cluster center ${\bf{c}}_r$ w.r.t. $C_r$. ${H_{\pi ,r}}: = \bigcup\nolimits_{t:\pi (t) = r} {{C_t}}$ denotes the union of clusters $\{ C_r \}$ with the closest distance to the $r$-th orthogonal base in ${\bf{O}}$. Let $A \Delta B: = (A \backslash B) \cup (B\backslash A)$ be the \textbf{symmetric difference} between sets $A$ and $B$. When $K$Means has an approximation ratio of $\alpha$, we can obtain
    \begin{equation*}
        \sum\nolimits_{r = 1}^K {\mu ({H_{\pi ,r}}\Delta {{\hat S}_r})}  \le 32(1 + \alpha ){(1 + {\lambda _1})^2}{\mu _{\max }}K{\Psi _{{\rm{ISC}}}}.
    \end{equation*}
\end{lemma}
The key idea to prove \textbf{Lemma~\ref{le2}} is to apply $\sum\nolimits_r {\sum\nolimits_{{v_i} \in {C_r}} {{a_{ir}}} }  \ge \sum\nolimits_r {\sum\nolimits_{t \ne \pi (r)} {\sum\nolimits_{{v_i} \in {C_r} \cap {{\hat S}_t}} {{a_{ir}}} } }$
(with $a_{ir}$ as a variable about $\{ i, r\}$) to the \textbf{clustering cost} (\ref{Eq:COST}) and then combine the result with \textbf{Lemma~\ref{le1}} and \textbf{Theorem~\ref{th2-1}}. Please refer to Appendix~\ref{App:Lemma2} for the full proof.

Based on \textbf{Lemma~\ref{le2}}, we derive our main theoretical result in the following \textbf{Theorem~\ref{Th:Main}}.
\begin{theorem}[\textbf{Main Theoretical Result}]\label{Th:Main}
    Given a graph $G$ and a pre-set number of clusters $K$, let $({\hat S}_1, \cdots, {\hat S}_K)$ be a $\bar{\phi}_K(G)$-\textbf{optimal} partition of \textbf{average conductance minimization} and $(C_1, \cdots, C_K )$ be the clustering result of \textit{ISC}, where the optimal correspondence of $C_r$ is $\hat S_r$. Assume that $K$Means has an approximation ratio of $\alpha$, which can be a constant \cite{choo2020k}. Let $\mathcal{M}$ be the set of mis-clustered nodes between $(C_1, \cdots, C_K )$ and $({\hat S}_1, \cdots, {\hat S}_K)$. We have
    \begin{equation*}
        |\mathcal{M}| \le 160(1 + \alpha ){(1 + {\lambda _1})^2}{\tilde \mu}K{\Psi _{{\rm{ISC}}}} = {c}{(1 + {\lambda _1})^2}{\tilde \mu}K{\Psi _{{\rm{ISC}}}},
    \end{equation*}
where $c := 160(1+\alpha)$ is a constant; $\tilde \mu := \mu_{\max} / d_{\min}$.
\end{theorem}
One can prove \textbf{Theorem~\ref{Th:Main}} by considering all possible relations between $\{ H_{\pi, r}\}$ and $\{ C_r \}$ similar to the proof of Theorem 2 in \cite{macgregor2022tighter} and using $d_{\min} |\mathcal{M}| \le \sum\nolimits_r {\mu ({C_r}\Delta {{\hat S}_r})}$. Appendix~\ref{App:Th-Main} gives the proof.

\textbf{Theorem~\ref{Th:Main}} provides an upper bound for the number of mis-clustered nodes $|\mathcal{M}|$ w.r.t. the \textbf{optimal solution} $({\hat S}_1, \cdots, {\hat S}_K)$ to \textbf{average conductance minimization}.
This bound is directly proportional to $\Psi_{\rm{ISC}}$. \textit{A graph with (\romannumeral1) higher degree heterogeneity and (\romannumeral2) weaker clustering structures causes a larger $\Psi_{\rm{ISC}}$ and thus a higher upper bound}.

To measure the quality of $(C_1, \cdots, C_K)$ given by a graph clustering algorithm, the \textbf{clustering accuracy} is defined as ${\mathop{\rm AC}\nolimits} (C_1, \cdots, C_K; \hat S_1, \cdots, \hat S_K) := 1 - |\mathcal{M}|/N$. Therefore, \textit{one can derive a lower bound for \textbf{clustering accuracy} based on the upper bound in \textbf{Theorem~\ref{Th:Main}}}. This lower bound is inversely proportional to $\Psi_{\rm{ISC}}$, which quantitatively reveals impacts of (\romannumeral1) degree heterogeneity and (\romannumeral2) weakness of clustering structures to the quality of DCSC. Concretely, \textit{higher degree heterogeneity and weaker clustering structures make $G$ harder to be clustered}. Our experiments also demonstrate that the quality of almost all the methods is significantly affected by these two aspects.
However, some related analyses may not be able to explicitly interpret such impacts.

Note that \textbf{Theorem~\ref{Th:Main}} does not rely on additional conditions, while most related theoretical results must satisfy some special conditions. For instance, Qin and Rohe \cite{RSC} proved that $|\mathcal{M}| \le {c_1}K\ln (N/\epsilon )/[{m^2}(\bar d_{\min}  + \tau )\lambda _K^2]$ for \textit{RSC} with probability at least $(1 - \epsilon)$, if $[K\ln (4N/\epsilon )/(\bar d_{\min}  + \tau )]^{1/2} \le {\lambda _K}/(8\sqrt 3 )$ and $(\bar d_{\min}  + \tau ) > 2\ln (4N/\epsilon )$ for sufficiently large $N$, where $c_1$ is a constant; $m$ is value related to DCSBM; $\bar d_{\min}$ is the minimum expected degree. The two conditions may fail for graphs with large $N$s, leading to trivial results. Moreover, this upper bound may be loose for a large $N$ and become trivial. Whereas, our bound in \textbf{Theorem~\ref{Th:Main}} is not related to $N$.

\begin{table}[]\scriptsize
\centering
\caption{Further theoretical results on the clustering quality of DCSC.}\label{Tab:DCSC}
\begin{tabular}{l|l|l|l|l}
\hline
\multirow{4}{*}{$\Psi$} & \textit{\textbf{NJW}} & $\Psi_{\rm{NJW}} := (1 - \lambda _{K + 1}^{{\rm{NJW}}} ) ^{-1} {\bar \phi _K} (G)$ & \multirow{4}{*}{$|\mathcal{M}| \le$} & $640(1 + \alpha ){\tilde \mu}K{\Psi _{{\rm{NJW}}}}$ \\ 
 & \textit{\textbf{RSC}} & $\Psi_{\rm{RSC}} := ( 1 - \lambda _{K + 1}^{\rm RSC})^{-1} [ {1 - {\tilde d} (1 - {{\bar \phi }_K} (G))} ]$ &  & $640(1 + \alpha ){\tilde \mu}K{\Psi _{{\rm{RSC}}}}$ \\ 
 & \textit{\textbf{SCORE+}} & ${\Psi _{{\rm{SC + }}}} = {(1 - \lambda _{K + 2}^{{\rm{SC + }}})^{ - 1}} [ {1 - \tilde d(1 - {{\bar \phi }_K} (G))} ]$ &  & $160(1 + \alpha ){(1 + \lambda _2^{{\rm{SC+}}})^2}{\tilde \mu}K{\Psi _{{\rm{SC+}}}}$ \\  
 & \textit{\textbf{ISC}} & $\Psi_{\rm{ISC}} := ( 1 - \lambda _{K + 2}^{{\rm{ISC}}} ) ^{-1} [ {1 - {\tilde d} (1 - {{\bar \phi }_K} (G))} ]$ &  & $160(1 + \alpha ){(1 + \lambda _1^{{\rm{ISC}}})^2}{\tilde \mu}K{\Psi _{{\rm{ISC}}}}$ \\ \hline
\end{tabular}
\end{table}

\subsection{General Theoretical Results for (DC)SC}
We further reduce our analysis on \textit{ISC} (i.e., \textbf{Theorem~\ref{Th:Main}}) to other (DC)SC algorithms and summarize their results in Table~\ref{Tab:DCSC}, where we use subscripts (superscripts) of `NJW', `RSC', `SC+', and `ISC' to denote variables of \textit{NJW}, \textit{RSC}, \textit{SCORE+}, and \textit{ISC}.

By comparing bounds of \textit{NJW} and \textit{RSC}, we introduce \textbf{Remark~\ref{Q1}} to answer \textbf{Q1}: \textbf{Can DCSC (e.g., \textit{RSC}) achieve tighter bounds than vanilla SC (e.g., \textit{NJW})}?
\begin{remark}\label{Q1}
    When a graph $G$ has a high degree heterogeneity and is not so well-clustered, \textit{RSC} may have a tighter bound of $|\mathcal{M}|$ than that of \textit{NJW}.
\end{remark}
\begin{proof}
    To prove \textbf{Remark~\ref{Q1}}, one can compare the upper bounds of $|\mathcal{M}|$ w.r.t. \textit{NJW} and \textit{RSC} (see Table~\ref{Tab:DCSC}), which is equivalent to comparing values of ${\Psi _{\rm{NJW}}}$ and ${\Psi _{\rm{RSC}}}$. When \textit{RSC} has a tighter upper bound, we have ${\Psi _{\rm{NJW}}} \ge {\Psi _{\rm{RSC}}} \Rightarrow ({\Psi _{\rm{NJW}}} - {\Psi _{\rm{RSC}}}) \ge 0$. We further obtain the following derivations:
\begin{align*}
    & {\Psi _{{\rm{NJW}}}} - {\Psi _{{\rm{RSC}}}} = \frac{1}{{1 - \lambda _{K + 1}^{{\rm{NJW}}}}}{{\bar \phi }_K (G)} - \frac{1}{{1 - \lambda _{K + 1}^{{\rm{RSC}}}}}[1 - \tilde d(1 - {{\bar \phi }_K (G)})] \ge 0, \\
    & \Rightarrow \frac{1}{{1 - \lambda _{K + 1}^{{\rm{NJW}}}}}{{\bar \phi }_K}(G) \ge \frac{1}{{1 - \lambda _{K + 1}^{{\rm{RSC}}}}}[1 - \tilde d(1 - {{\bar \phi }_K}(G))] = \frac{{1 - \tilde d}}{{1 - \lambda _{K + 1}^{{\rm{RSC}}}}} + \frac{1}{{1 - \lambda _{K + 1}^{{\rm{RSC}}}}}\tilde d{{\bar \phi }_K}(G), \\
    & \Rightarrow [\frac{1}{{1 - \lambda _{K + 1}^{{\rm{NJW}}}}} - \frac{{\tilde d}}{{1 - \lambda _{K + 1}^{{\rm{RSC}}}}}]{{\bar \phi }_K} (G) \ge \frac{{1 - \tilde d}}{{1 - \lambda _{K + 1}^{{\rm{RSC}}}}}, \\
    & \Rightarrow [\frac{{1 - \lambda _{K + 1}^{{\rm{RSC}}}}}{{1 - \lambda _{K + 1}^{{\rm{NJW}}}}} - \tilde d]{{\bar \phi }_K} (G) \ge 1 - \tilde d, \\
    & \Rightarrow [\frac{{1 - \lambda _{K + 1}^{{\rm{RSC}}}}}{{1 - \lambda _{K + 1}^{{\rm{NJW}}}}} - \frac{{{d_{\min }}}}{{{d_{\max }} + \tau }}]{{\bar \phi }_K} (G) \ge \frac{{{d_{\max }} - {d_{\min }} + \tau }}{{{d_{\max }} + \tau }}, \\
    & \Rightarrow \frac{{1 - \lambda _{K + 1}^{{\rm{RSC}}}}}{{1 - \lambda _{K + 1}^{{\rm{NJW}}}}}\frac{{{d_{\max }} + \tau }}{{{d_{\max }} - {d_{\min }} + \tau }} - \frac{{{d_{\min }}}}{{{d_{\max }} - {d_{\min }} + \tau }} \ge \bar \phi _K^{ - 1} (G).
\end{align*}
Let $q: = (1 - \lambda _{K + 1}^{{\rm{RSC}}})/(1 - \lambda _{K + 1}^{{\rm{NJW}}})$. Usually, we have $\lambda_{K+1}^{\rm{RSC}} \le \lambda_{K+1}^{\rm{NJW}}$ and thus $q \ge 1$. Assume \textit{RSC} adopts its default setting of $\tau$ (i.e., $\tau = \bar d$). One can rewrite the inequality as
\begin{align*}
    \frac{{q ({d_{\max }} + \tau ) - {d_{\min }}}}{{{d_{\max }} - {d_{\min }} + \tau }} & = \frac{{q {d_{\max }} - {d_{\min }} + q \bar d}}{{{d_{\max }} - {d_{\min }} + \bar d}} \ge \bar \phi _K^{ - 1}(G).
\end{align*}
To ensure that the inequality holds, one may first ensure that the right part $\bar \phi_K^{-1} (G)$ is small enough (i.e., $\phi_K (G)$ is large). It implies that \textbf{the graph $G$ is not so well-clustered}, in contrast to the well-clustered condition \cite{NJW,Mizutani}. Moreover, one may also ensure that the left part is large enough. With the increase of degree heterogeneity, the numerator increases faster than denominator. Therefore, \textbf{higher degree heterogeneity results in a larger value of the left part}.

In summary, when a graph $G$ (\romannumeral1) has a high degree heterogeneity and (\romannumeral2) is not so well-clustered, \textit{RSC} may have a tighter upper bound of $|\mathcal{M}|$ than that of \textit{NJW}.
\end{proof}

It is consistent with conclusions of prior DCSBM-based analyses \cite{RSC,zhang2018understanding} that \textit{RSC} may outperform \textit{NJW} when the degree heterogeneity is high. In early studies on vanilla SC \cite{NJW}, a graph is defined to be \textbf{well-clustered}, if $\bar \phi_K(G) / (1 - \lambda_{K+1})$ is sufficiently small, consistent with that $\bar \phi_K (G)$ is small. This well-clustered assumption implies that the optimal solution $(\hat S_1, \cdots, \hat S_K)$ describes an explicit clustering structure of $G$.

We then introduce the following \textbf{Remark~\ref{Q2}} to answer \textbf{Q2}: \textbf{Can \textit{ISC} achieve a tighter bound than \textit{RSC}}?
\begin{remark}\label{Q2}
  Suppose that \textit{RSC} and \textit{ISC} (\romannumeral1) have almost the same $\alpha$ for $K$Means and (\romannumeral2) use the same correction $\tau$. Then, \textit{ISC} always has a tighter bound of $|\mathcal{M}|$ than that of \textit{RSC}.
\end{remark}
\begin{proof}
    Suppose that \textit{RSC} and \textit{ISC} have almost the same approximation ratio $\alpha$ of $K$Means. Moreover, suppose that \textit{RSC} and \textit{ISC} use the same correction term $\tau$. Then, we have
\begin{equation*}
    \lambda _{K + 1}^{\rm{RSC}} = \lambda _{K + 1}^{\rm{ISC}} \ge \lambda _{K + 2}^{\rm{ISC}} \Rightarrow {\Psi _{{\rm{RSC}}}} \ge {\Psi _{{\rm{ISC}}}}.
\end{equation*}
Note that $\lambda_1^{\rm{ISC}} < 1$. For the upper bound of $|\mathcal{M}|$ (see Table~\ref{Tab:DCSC}), we further have
\begin{equation*}
    160(1 + \alpha ){(1 + \lambda _1^{{\rm{ISC}}})^2}{\tilde \mu}K{\Psi _{{\rm{ISC}}}} < 640(1 + \alpha ){\tilde \mu}K{\Psi _{{\rm{RSC}}}},
\end{equation*}
which indicates that \textit{ISC} has a tighter upper bound of $|\mathcal{M}|$ than that of \textit{RSC}.
\end{proof}

\section{Extension of DCSC: ASCENT}\label{Sec:Meth}
Inspired by recent advances in GNNs, we introduce ASCENT, a simple yet effective extension of DCSC.
Different from most DCSC methods with a constant correction $\tau$ (e.g., \textit{RSC}, \textit{SCORE+}, and \textit{ISC} in Table~\ref{Tab:Alg-Sum}), ASCENT adopts a node-wise correction scheme. It can determine different corrections $\{ \tau_i \}$ for nodes $\{ v_i \}$ via an iterative aggregation mechanism that computes `local' average degrees w.r.t. graph topology.
Whereas, $\tau$ is usually set to be a `global' average degree for existing DCSC algorithms (e.g., $\tau = \bar d$ for \textit{RSC}).

\subsection{The Presented Algorithm}
Let ${\bm{\tau}}^{(l)} \in \mathbb{R}^N_+$ be the vector of node-wise corrections in the $l$-th iteration, with $\tau_i^{(l)}$ as the correction of node $v_i$. We let ${{\bm{\tau }}^{(0)}} = {\bf{d}}$ (with $\tau_i^{(0)} = d_i$) for initialization. Suppose that there are in total $L$ iterations. We obtain the node-wise corrections ${\bm{\tau}} \in \mathbb{R}^N_+$ via
\begin{equation}\label{Eq:tau}
    {{\bm{\tau }}^{(l)}} := {{\bf{\hat D}}^{ - 1}}{\bf{\hat A}}{{\bm{\tau }}^{(l - 1)}}~(1 \le l \le L),{\rm{~and~}}{\bm{\tau }}: = \theta {{\bm{\tau }}^{(L)}},
\end{equation}
where ${\bf{\hat A}} := {\bf{A}} + {\bf{I}}_N$ is the adjacency matrix with self-edges; ${\bf{\hat D}}$ is the degree diagonal matrix w.r.t. ${\bf{\hat A}}$; $\theta >0$ is a hyper-parameter.

In (\ref{Eq:tau}), we iteratively update ${\bm{\tau}}^{(l)}$ using the \textbf{GNN mean aggregation} \cite{hamilton2017inductive}.
Different from existing GNN-based methods \cite{bianchi2020spectral,tsitsulin2023graph,bhowmick2024dgcluster}, ASCENT does not rely on any attribute inputs or training procedures. Instead, it uses $\{ \bm{\tau}^{(l)} \in \mathbb{R}^{N}_+ \}$ as features for aggregation. In each iteration, it computes the average correction value w.r.t. the one-hop neighbors for each node.
We use $\tau_i = \theta \tau_i^{(L)}$ as the final correction of node $v_i$. Similar to the role of $\delta$ in \textit{RSC}, \textit{SCORE+}, and \textit{ISC} summarized in Table~\ref{Tab:Alg-Sum}, $\theta$ adjusts the scale of $\tau_i$.
Then, ASCENT adopts the same strategies of spectral embedding arrangement and normalization (i.e., steps (\uppercase\expandafter{\romannumeral2}) and (\uppercase\expandafter{\romannumeral3}) in Table~\ref{Tab:Alg-Sum}) as \textit{ISC}.
Algorithm~\ref{Alg:ASCENT} summarizes the overall procedure of ASCENT.

\begin{algorithm}[t]\footnotesize
\caption{\footnotesize The Proposed \textbf{ASCENT} Algorithm}
\label{Alg:ASCENT}
\KwIn{graph $G = (V, E)$, number of clusters $K$, hyper-parameters $\{ \theta, L \}$}
\KwOut{a feasible clustering result $(C_1, \cdots, C_K)$}
{
$\rhd$ \textbf{Construct regularized graph Laplacian} ${\bf{L}}_{\tau}$ \\
}
\For{{\bf{each}} node $v_i \in V$}
{
    $\tau_i^{(0)} \leftarrow d_i$ //Initialize node-wise corrections
}
\For{$l$ {\bf{from}} $1$ {\bf{to}} $L$}
{
    ${{\bm{\tau }}^{(l)}} \leftarrow {{\bf{\hat D}}^{ - 1}}{\bf{\hat A}} {{\bm{\tau }}^{(l-1)}}$ //Iteratively update node-wise corrections
}
{
${\bm{\tau}} \leftarrow \theta {\bm{\tau}}^{(L)}$ //Final node-wise corrections\\
${{\bf{L}}_\tau } \leftarrow {({\bf{D}} + {\mathop{\rm diag}\nolimits} ({\bm{\tau }}))^{ - 1/2}}{\bf{A}}{({\bf{D}} + {\mathop{\rm diag}\nolimits} ({\bm{\tau }}))^{ - 1/2}}$\\
$\rhd$ \textbf{Step (\uppercase\expandafter{\romannumeral1})}: \textbf{ED on graph Laplacian}\\
Find the leading $(K+1)$ eigenvalues $(\lambda_1, \cdots, \lambda_{K+1})$ and eigenvectors $({\bf{u}}_1, \cdots, {\bf{u}}_{K+1})$ of ${{\bf{L}}_\tau }$\\
$\rhd$ \textbf{Step (\uppercase\expandafter{\romannumeral2})}: \textbf{Arrangement of spectral embedding}\\
${\bf{F}} \leftarrow [{\lambda _1}{{\bf{u}}_1}, \cdots ,{\lambda _{K + 1}}{{\bf{u}}_{K + 1}}]$\\
$\rhd$ \textbf{Step (\uppercase\expandafter{\romannumeral3})}: \textbf{Normalization of arranged embedding}\\
}
\For{{\bf{each}} node $v_i \in V$}
{
    ${{\bf{\tilde F}}_{i,:}} \leftarrow {{\bf{F}}_{i,:}}/|{{\bf{F}}_{i,:}}{|_2}$
}
{
$\rhd$ \textbf{Step (\uppercase\expandafter{\romannumeral4})}: \textbf{$K$Means clustering}\\
apply $K$Means to rows of ${\bf{\tilde F}}$ to get the clustering result $(C_1, \cdots, C_K)$
}
\end{algorithm}

Fig.~\ref{Fig:Case} demonstrates our node-wise correction scheme on the \textbf{Zachary's karate club} graph \cite{zachary1977information} with $2$ clusters, where we visualize the normalized $\{ {\bm{\tau}}^{(l)} \}$ in different iterations; each color denotes a cluster. Although different nodes have various initial values (i.e., degrees) in ${\bm{\tau}}^{(0)}$, the mean aggregation in (\ref{Eq:tau}) \textit{helps nodes in the same cluster (i.e., with more common neighbors) to have close correction values}. For instance, when $30 \le l \le 70$, nodes in the first cluster tend to have larger corrections than those in the second cluster.
It is well-known that most GNNs suffer from over-smoothing \cite{rusch2023survey}, where node features converge to a constant as the number of layers increases. Similarly, the node-wise corrections of ASCENT also converge to a constant for a large number of iterations $l$ (e.g., ${\bm{\tau}}^{(100)}$), due to over-smoothing (i.e., $\mathop {\lim }_{l \to \infty } \tau _i^{(l)} = c, \forall {v_i} \in V$, with $c$ as a constant). In this case, \textit{ASCENT reduces to existing DCSC methods with a constant correction $\tau$}, corresponding to a `global' average of degrees. Our experiments indicate that \textit{ASCENT can potentially achieve better clustering quality in some early stages before over-smoothing}.
It corresponds to a special `local' average of degrees.

\begin{figure}
\centering
 \begin{minipage}{0.12\linewidth}
 \subfigure[Topology]{
  \includegraphics[width=\textwidth,trim=0 0 0 0,clip]{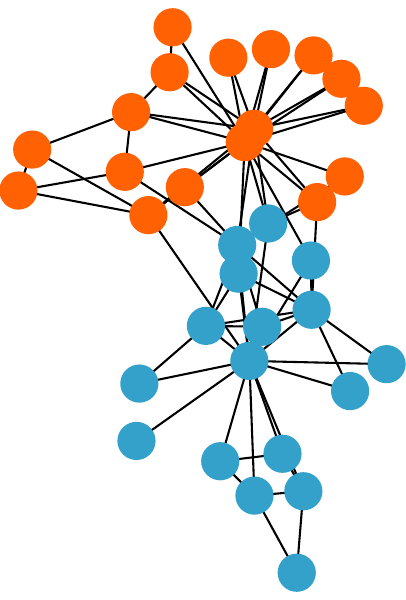}}
 \end{minipage}
 \begin{minipage}{0.13\linewidth}
 \subfigure[${\bm{\tau}}^{(0)}$]{
  \includegraphics[width=\textwidth,trim=42 15 25 19,clip]{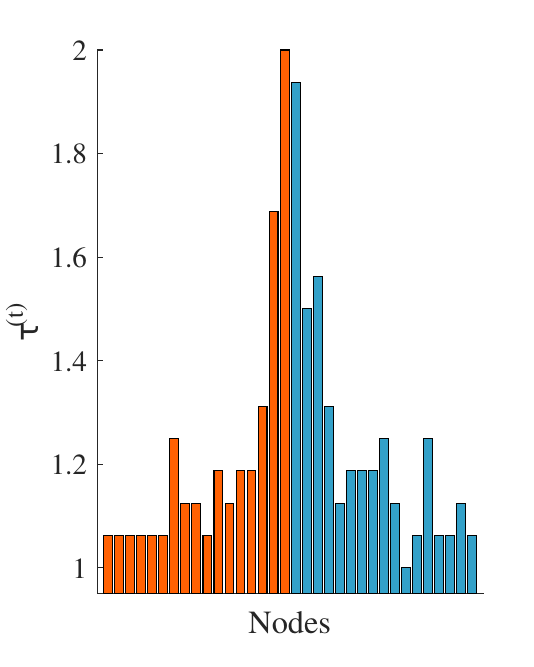}}
 \end{minipage}
 \begin{minipage}{0.13\linewidth}
 \subfigure[${\bm{\tau}}^{(10)}$]{
  \includegraphics[width=\textwidth,trim=42 15 25 19,clip]{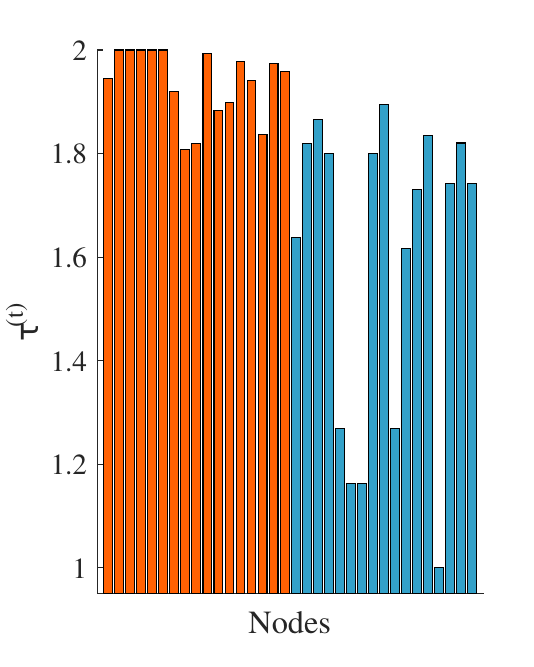}}
 \end{minipage}
 \begin{minipage}{0.13\linewidth}
 \subfigure[${\bm{\tau}}^{(30)}$]{
  \includegraphics[width=\textwidth,trim=42 15 25 19,clip]{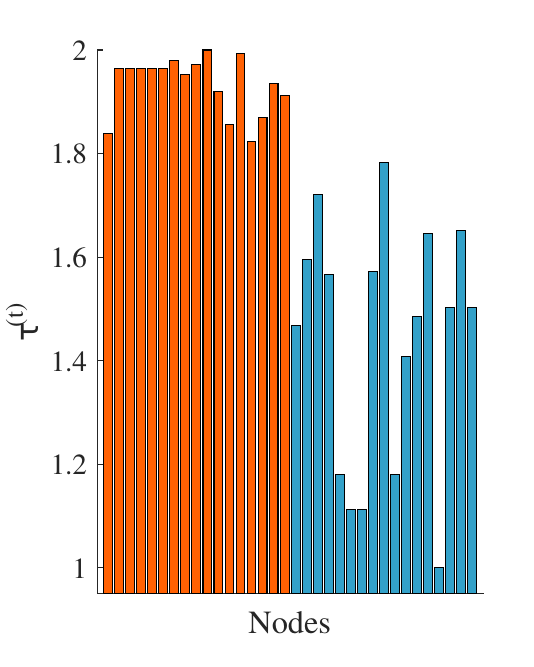}}
 \end{minipage}
 \begin{minipage}{0.13\linewidth}
 \subfigure[${\bm{\tau}}^{(50)}$]{
  \includegraphics[width=\textwidth,trim=42 15 25 19,clip]{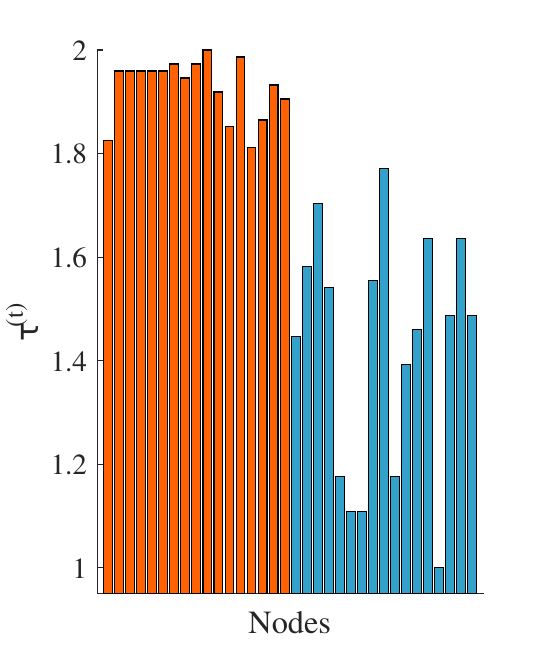}}
 \end{minipage}
 \begin{minipage}{0.13\linewidth}
 \subfigure[${\bm{\tau}}^{(70)}$]{
  \includegraphics[width=\textwidth,trim=42 15 25 19,clip]{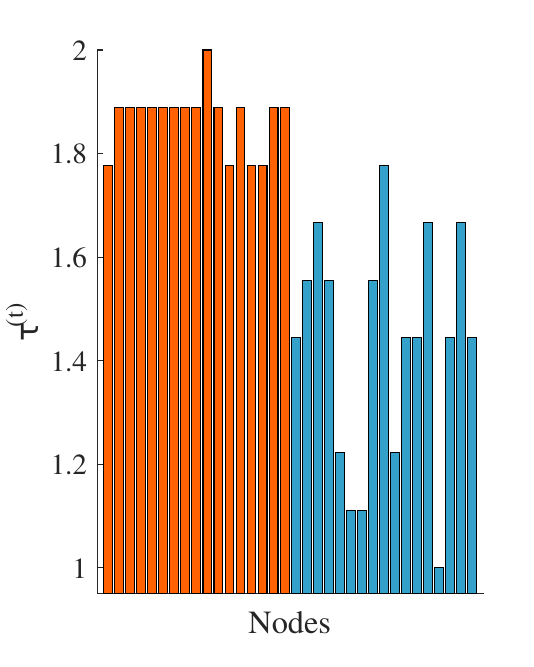}}
 \end{minipage}
 \begin{minipage}{0.13\linewidth}
 \subfigure[${\bm{\tau}}^{(100)}$]{
  \includegraphics[width=\textwidth,trim=42 15 25 19,clip]{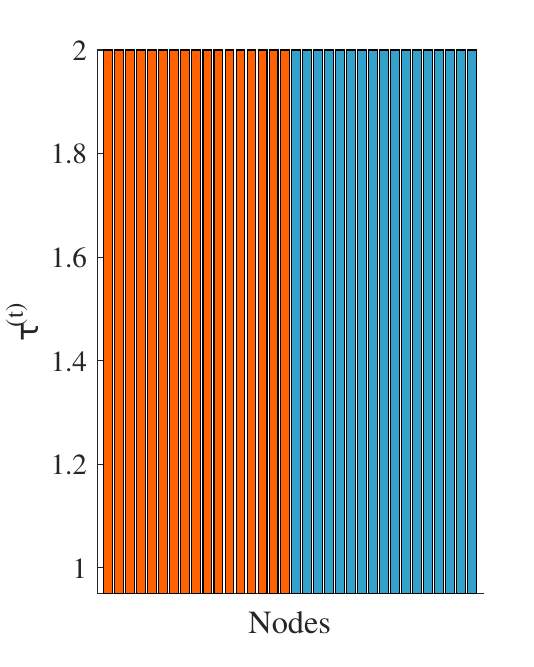}}
 \end{minipage}
\caption{$\{ \bm{\tau}^{(l)} \}$ on the \textbf{Zachary's karate club} graph, where each color denotes a cluster.}\label{Fig:Case}
\end{figure}

\subsection{Extended Theoretical Analysis}
We extend our analysis to the following \textbf{Proposition~\ref{Th:ASCENT}} for ASCENT (see Appendix~\ref{App:Th-ASCENT} for the proof), where we use `AST' to denote variables of ASCENT.
For each cluster $\hat S_r$, we introduce a cluster-wise correction ${\hat \tau _r}: = \max \{ {\tau _i}|{v_i} \in {\hat S_r}\}$. Since different nodes $\{ v_i \}$ may have different $\{ \tau_i \}$, different clusters $\{ \hat S_r \}$ may also have different $\{\hat \tau_r\}$, which can verify advantages of node-wise corrections $\{ \tau_i\}$ over conventional methods with a constant $\tau$.
\begin{proposition}\label{Th:ASCENT}
    Let ${\bar \varphi _K}(G): = \sum\nolimits_{r = 1}^K {{{\tilde d}_r}} \phi ({\hat S_r})/K$ be a \textbf{reweighted conductance} w.r.t. the optimal partition $({\hat S_1}, \cdots ,{\hat S_K})$, where ${\tilde d_r}: = {d_{\min }}/({d_{\max }} + {\hat \tau _r})$. By replacing $\Psi_{\rm{ISC}}$ and $\lambda_r^{\rm{ISC}}$ with $\Psi_{\rm{AST}} : = {(1 + {\lambda _{K + 2}^{\rm{AST}}})^{ - 1}}[1 - (\hat d - {\bar \varphi _K}(G))]$ and $\lambda_r^{\rm{AST}}$, where $\hat d: = \sum\nolimits_{r = 1}^K {{{\tilde d}_r}} /K$, \textbf{Theorem~\ref{Th:Main}} holds for \textbf{ASCENT}.
\end{proposition}

By comparing bounds of \textbf{Proposition~\ref{Th:ASCENT}} and \textbf{Theorem~\ref{Th:Main}}, we answer \textbf{Q3}: \textbf{Can ASCENT achieve a tighter bound than \textit{ISC}}?
The following \textbf{Remark~\ref{Q3}} demonstrates that \textit{it is possible for ASCENT to obtain a tighter bound}.
\begin{remark}\label{Q3}
    Let $\tau$ be the constant correction of \textit{ISC}. ASCENT has a tighter bound of $|\mathcal{M}|$ than \textit{ISC}, when $\lambda_{K+2}^{\rm{ISC}} \ge \lambda_{K+2}^{\rm{AST}}$ and $\sum\nolimits_{r = 1}^K {(\tau  - {{\hat \tau }_r})(1 - \phi ({{\hat S}_r}))/({d_{\max }} + {{\hat \tau }_r})}  \ge 0$, which are possible to satisfy.
\end{remark}
\begin{proof}
    When $\lambda_{K+2}^{\rm{ISC}} \ge \lambda_{K+2}^{\rm{AST}}$, we have
\begin{align*}
    {\Psi _{{\rm{ISC}}}} = {(1 - \lambda _{K + 2}^{{\rm{ISC}}})^{ - 1}}[1 - (\tilde d - \tilde d{{\bar \phi }_K}(G))] \ge {(1 - \lambda _{K + 2}^{{\rm{AST}}})^{ - 1}}[1 - (\tilde d - \tilde d{{\bar \phi }_K}(G))].
\end{align*}
One can further derive
\begin{align*}
    {\Psi _{{\rm{ISC}}}} - {\Psi _{{\rm{AST}}}} & \ge \frac{1}{{1 - \lambda _{K + 2}^{{\rm{AST}}}}}[(\hat d - \tilde d) + (\tilde d{{\bar \phi }_K}(G) - {{\bar \varphi }_K}(G))] \\
    & = \frac{1}{{1 - \lambda _{K + 2}^{{\rm{AST}}}}}\frac{1}{K}\sum\limits_{r = 1}^K {\left[ {(\frac{{{d_{\min }}}}{{{d_{\max }} + {{\hat \tau }_r}}} - \frac{{{d_{\min }}}}{{{d_{\max }} + \tau }})(1 - \phi ({{\hat S}_r}))} \right]} \\
    & = \frac{1}{{1 - \lambda _{K + 2}^{{\rm{AST}}}}}\frac{1}{K}\sum\limits_{r = 1}^K {\left[ {\frac{{{d_{\min }}(\tau  - {{\hat \tau }_r})}}{{({d_{\max }} + \tau )({d_{\max }} + {{\hat \tau }_r})}}(1 - \phi ({{\hat S}_r}))} \right]} \\
    & = \frac{1}{{1 - \lambda _{K + 2}^{{\rm{AST}}}}}\frac{{\tilde d}}{K}\sum\limits_{r = 1}^K {\left[ {\frac{{\tau  - {{\hat \tau }_r}}}{{{d_{\max }} + {{\hat \tau }_r}}}(1 - \phi ({{\hat S}_r}))} \right]}.
\end{align*}
To ensure ${\Psi _{{\rm{ISC}}}} \ge {\Psi _{{\rm{AST}}}} \Rightarrow ( {\Psi _{{\rm{ISC}}}} - {\Psi _{{\rm{AST}}}} ) \ge 0$, one needs to ensure
\begin{equation*}
    \sum\limits_{r = 1}^K {\frac{{\tau  - {{\hat \tau }_r}}}{{{d_{\max }} + {{\hat \tau }_r}}}} (1 - \phi ({\hat S_r})) \ge 0.
\end{equation*}
In particular, it is possible for ASCENT to satisfy the aforementioned two conditions.
\end{proof}

\subsection{Complexity Analysis}
Given a large-scale graph, we usually have $K, L \ll N < |E|$. Assume that the graph to be partitioned is sparse. For ASCENT, the time complexity of deriving node-wise corrections $\{ \tau_i \}$ is no more than $O(|E|L) = O(|E|)$ by fully utilizing the sparsity of a graph and sparse-dense matrix multiplication operation. ASCENT follows the same steps of (\romannumeral1) ED, (\romannumeral2) spectral embedding arrangement, (\romannumeral3) embedding normalization, and (\romannumeral4) $K$Means clustering with \textit{ISC}, which have complexities of (\romannumeral1) $O((N + |E|) K) = O(N + |E|)$ (e.g., via the efficient Lanczos algorithm \cite{lehoucq1998arpack} for ED), (\romannumeral2) $O(NK) = O(N)$, (\romannumeral3) $O(NK) = O(N)$, and (\romannumeral4) $O(NK^2t) = O(N)$ (with $t \ll N$ as the number of iterations in $K$Means), respectively.

In summary, the overall time complexity of ASCENT is about $O(N + |E|)$. It has the same complexity as most existing DCSC algorithms. Therefore, \textit{the additional step of deriving node-wise corrections $\{ \tau_i \}$ will not increase the complexity of ASCENT}.

\section{Experiments}\label{Sec:Exp}

\begin{table}[]\scriptsize
\caption{Statistic details of datasets.}\label{Tab:Data}
\centering
\begin{tabular}{l|l|l|l|lll}
\hline
\textbf{Datasets} & $N$ & $|E|$ & $K$ & min & max & \text{avg~$d$} \\ \hline
LFR-1 & 2,000 & 7,693-12,057 & 2-15 & \text{1-2} & \text{286-1,000} & \text{7-12} \\
LFR-2 & 2,000 & 8,489-32,202 & 4-14 & \text{2-6} & \text{375-1,000} & \text{8-32} \\ \hline
Caltech & 590 & 12,822 & 8 & 1 & 179 & 44 \\
Simmons & 1,137 & 24,257 & 4 & 1 & 293 & 43 \\
PolBlogs & 1,222 & 16,714 & 2 & 1 & 351 & 27 \\
BioGrid & 5,640 & 59,748 & 81 & 1 & 2570 & 21 \\ \hline
PPI & 3852 & 37841 & N/A & 1 & 593 & 12 \\
Wiki & 4,777 & 92,295 & N/A & 2 & 3,644 & 39 \\
BlogCatalog & 10,312 & 333,983 & N/A & 1 & 3,992 & 64 \\
\text{ogb-Protein} & 132,534 & 39,561,252 & N/A & 1 & 7,750 & 597 \\ 
RoadCA & 1,957,027 & 2,760,388 & N/A & 1 & 12 & 3 \\ 
LiveJournal & 3,997,962 & 34,681,189 & N/A & 1 & 14,815 & 17 \\ 
\text{ogb-Products} & 2,385,629 & 61,789,416 & N/A & 1 & 17,480 & 52 \\ 
Orkut & 3,072,441 & 117,185,083 & N/A & 1 & 33,313 & 76 \\ \hline
\end{tabular}
\end{table}

\begin{table}[]\scriptsize
\caption{Summary of methods to be evaluated.}\label{Tab:Meth}
\centering
\begin{tabular}{l|l}
\hline
\textbf{Baselines} & \textbf{Venues} \\ \hline
NJW & NeurIPS 2001~\cite{NJW}\\
SCORE & Annals of Statistics 2015~\cite{jin2015fast}\\
RSC & NeurIPS 2013~\cite{RSC}\\
SCORE+ (SC+) & Sankhya A 2021~\cite{jin2021improvements}\\
ISC & arXiv 2020~\cite{ISC}\\ 
IBHCD & NeurIPS 2019~\cite{dall2019revisiting}\\ 
CDBH & JMLR 2021~\cite{dall2021unified}\\ \hline
GraphEncoder (GE) & AAAI 2014~\cite{tian2014learning}\\
SDCN & WWW 2020~\cite{bo2020structural}\\
MinCutPool (MCP) & ICML 2020~\cite{bianchi2020spectral}\\
DMoN & JMLR 2023~\cite{tsitsulin2023graph}\\
S3GC & NeurIPS 2022~\cite{devvrit2022s3gc} \\ 
MAGI & KDD 2024~\cite{liu2024revisiting}\\\hline
\textbf{ASCENT} (\textbf{AST}) & \textbf{Ours}\\\hline
\end{tabular}
\end{table}

\subsection{Experiment Setup}
\textbf{Datasets}.
We used $2$ settings of synthetic benchmarks and $12$ real graphs for evaluation. Table~\ref{Tab:Data} summarizes their statistics, where $N$, $|E|$, and $K$ are the numbers of nodes, edges, and clusters (if available); $d$ denotes node degree.
In particular, these datasets cover various scales, with $N$ increasing from about $500$ to more than $3 \times 10^6$, which can effectively test the scalability of a graph clustering method.

LFR-net \cite{lancichinetti2008benchmark} is a synthetic benchmark that simulates properties of real-world graphs. It uses $({\bar d},{d_{\max }},{c_{\min }},{c_{\max }},\eta )$ to generate a graph, where ${\bar d}$ and ${d_{\max}}$ are the average and maximum degrees; ${c_{\min }}$ and ${c_{\max }}$ are the minimum and maximum cluster sizes; $\eta$ is the ratio between the external degree and total degree of a node $v_i$ w.r.t. the cluster $v_i$ belongs to.
To test the ability to handle (\romannumeral1) weak clustering structures and (\romannumeral2) high degree heterogeneity, we generated two sets of graphs, denoted as \textbf{LFR-1} and \textbf{LFR-2}.
For \textbf{LFR-1}, we fixed $(N, {\bar d}, d_{\max}, c_{\min}, c_{\max}) = (2000, 50, 1000, 50, 500)$ and adjusted $\eta \in \{ 0.1, 0.3, 0.5\}$. \textit{With the increase of $\eta$, clustering structures become weaker}. For \textbf{LFR-2}, we fixed $(N, d_{\max}, c_{\min}, c_{\max}, \eta) = (2000, 1000, 50, 500, 0.5)$ and set ${\bar d} \in \{ 10, 20, 30\}$, where \textit{lower $\bar d$ implies higher degree heterogeneity}.

\textbf{Caltech} \cite{red2011comparing} and \textbf{Simmons} \cite{red2011comparing} are two graphs regarding friendships of two online social networks.
\textbf{PolBlogs} \cite{adamic2005political} is a graph constructed based on the links between blogs with different political leanings.
\textbf{Wiki} \cite{grover2016node2vec} is a co-occurrence graph of words that appear in the first million bytes of the Wikipedia dump.
\textbf{BlogCatalog} \cite{grover2016node2vec} was extracted from social relationships provided by blogger authors.
\textbf{PPI} \cite{grover2016node2vec}, \textbf{BioGrid} \cite{stark2006biogrid}, and \textbf{ogbn-Protein} \cite{szklarczyk2019string} are three protein-protein interaction graphs. \textbf{LiveJournal} \cite{yang2012defining} and \textbf{Orkut} \cite{yang2012defining} were constructed based on the friendship relations in the online social networks of LiveJournal and Orkut. \textbf{RoadCA} \cite{leskovec2009community} describes intersections and endpoints of a road network in California. \textbf{ogbn-Products} \cite{chiang2019cluster} is a graph representing an Amazon product co-purchasing network.
During preprocessing, we followed \cite{qin2010spectral} to extract the clustering ground-truth of \textbf{BioGrid}, where the complex set CYC2008 \cite{pu2009up} with $231$ protein complexes was used as the reference set. For the rest datasets, we used their original formats.

Note that we could not use ground-truth of \textbf{PPI}, \textbf{Wiki}, \textbf{BlogCatalog}, \textbf{LiveJournal}, and \textbf{Orkut}, which describe overlapping community structures, as we consider disjoint graph partitioning.
As stated in Section~\ref{Sec:Prob}, we assume that graph attributes are unavailable. \textbf{ogbn-Protein} and \textbf{ogbn-Products} are attributed graphs for the evaluation of node classification. We did not use their ground-truth, because it is unclear that such ground-truth is dominated by topology or attributes.
\textbf{RoadCA} does not provide ground-truth.

\textbf{Baselines}.
As in Table~\ref{Tab:Meth}, we compared ASCENT over $13$ baselines with $2$ categories.
First, \textit{NJW} \cite{NJW}, \textit{SCORE} \cite{jin2015fast}, \textit{RSC} \cite{RSC}, \textit{SCORE+} \cite{jin2021improvements}, \textit{ISC} \cite{ISC}, \textit{IBHCD} \cite{dall2019revisiting}, and \textit{CDBH} \cite{dall2021unified} are SC methods.
Second, \textit{GraphEncoder} (\textit{GE}) \cite{tian2014learning}, \textit{SDCN} \cite{bo2020structural}, \textit{MinCutPool} (\textit{MCP}) \cite{bianchi2020spectral}, \textit{DMoN} \cite{tsitsulin2023graph}, \textit{S3GC} \cite{devvrit2022s3gc}, and \textit{MAGI} \cite{liu2024revisiting} are DGC approaches.

In particular, \textit{RSC}, \textit{SCORE+}, \textit{ISC}, \textit{IBHCD}, and \textit{CDBH} are DCSC baselines, where \textit{IBHCD} and \textit{CDBH} follow a new SC framework of iteratively searching each leading eigenvalue of the Bethe-Hessian matrix instead of that shown in Table~\ref{Tab:Alg-Sum}.
Note that we consider graph clustering without attributes, while most of the second type of baselines were originally designed for attributed graphs. We tried widely-used strategies of (\romannumeral1) SVD on the adjacency matrix and (\romannumeral2) one-hot encoding of degrees to derive feature inputs for them, with the best quality metrics reported.

\textbf{Evaluation Metrics}.
For datasets with feasible ground-truth, we used \textbf{normalized mutual information} (\textbf{NMI}) and \textbf{accuracy} (\textbf{AC}) as quality metrics.
We also used the \textbf{average conductance} defined in (\ref{Eq:Cond}) as an unsupervised metric for all the datasets.
For those without feasible ground-truth, we followed prior work \cite{dhillon2007weighted,chan2011convex} to set $K \in \{ 2, 8, 32\}$ and recorded the corresponding average conductance values. Usually, \textit{smaller conductance as well as larger NMI and AC indicate better clustering quality}.
We define that a method encounter the out-of-time exception if it fails to derive a feasible result within $10^4$ seconds.
Hyper-parameters of all the methods were tuned based on the unsupervised average conductance metric.

\begin{table}[]\scriptsize
\caption{Synthetic graph analysis on \textbf{LFR-1} and \textbf{LFR-2}, with the best and second-best metrics in \textbf{bold} and \underline{underlined}.}\label{Tab:Eva-Res-LFR}
\centering
\begin{tabular}{p{0.75cm}|p{0.35cm}p{0.35cm}p{0.55cm}|p{0.35cm}p{0.35cm}p{0.55cm}|p{0.35cm}p{0.35cm}p{0.55cm}|p{0.35cm}p{0.35cm}p{0.55cm}|p{0.35cm}p{0.35cm}p{0.55cm}|p{0.35cm}p{0.35cm}p{0.55cm}}
\hline
\multirow{2}{*}{} & \multicolumn{3}{c|}{\textbf{LFR-1}:$\eta$=0.5} & \multicolumn{3}{c|}{0.3} & \multicolumn{3}{c|}{0.1} & \multicolumn{3}{c|}{\textbf{LFR-2}:$d$=10} & \multicolumn{3}{c|}{20} & \multicolumn{3}{c}{30} \\ \cline{2-19} 
 & NMI & AC & Cond & NMI & AC & Cond & NMI & AC & Cond & NMI & AC & Cond & NMI & AC & Cond & NMI & AC & Cond \\ 
& ($\uparrow$,\%) & ($\uparrow$,\%) & ($\downarrow$,\%) & ($\uparrow$,\%) & ($\uparrow$,\%) & ($\downarrow$,\%) & ($\uparrow$,\%) & ($\uparrow$,\%) & ($\downarrow$,\%) & ($\uparrow$,\%) & ($\uparrow$,\%) & ($\downarrow$,\%) & ($\uparrow$,\%) & ($\uparrow$,\%) & ($\downarrow$,\%) & ($\uparrow$,\%) & ($\uparrow$,\%) & ($\downarrow$,\%) \\ \hline
NJW & 24.91 & 43.93 & 57.12 & 56.90 & 77.06 & 34.89 & 87.64 & 96.90 & 10.04 & 26.27 & 44.32 & 57.10 & 66.58 & 79.15 & 53.18 & 81.47 & 90.33 & 51.37 \\
SCORE & 30.19 & 48.60 & 66.10 & 59.77 & 80.17 & 39.64 & 77.15 & 86.90 & 21.68 & 30.09 & 47.74 & 66.48 & 65.19 & 77.27 & 57.22 & 83.54 & 89.17 & 53.55 \\
RSC & 39.18 & 60.59 & 52.82 & 54.43 & 78.65 & 32.17 & 80.73 & 94.98 & 10.11 & 39.75 & 61.22 & 53.94 & 61.86 & 77.19 & 52.60 & 76.63 & 87.41 & 51.66 \\
SC+ & 43.66 & 65.38 & 53.94 & 60.19 & 84.03 & 32.32 & 78.99 & 91.88 & 12.75 & 44.72 & 66.30 & 53.49 & 71.81 & 85.38 & 52.22 & 85.64 & 93.61 & 50.99 \\
ISC & \underline{44.16} & \underline{65.73} & \underline{52.31} & 65.27 & \underline{86.59} & \underline{31.79} & 88.01 & 97.23 & 10.00 & \underline{44.77} & \underline{66.59} & \underline{52.87} & \underline{72.62} & \underline{85.59} & 52.21 & \underline{86.88} & \underline{93.91} & 50.97 \\
IBHCD & 37.60 & 53.97 & 58.52 & 55.59 & 78.61 & 32.04 & 83.25 & 95.86 & \underline{9.86} & 38.27 & 55.39 & 58.41 & 69.85 & 79.32 & \underline{51.70} & 81.93 & 88.28 & 50.94 \\
CDBH & 39.06 & 59.40 & 54.36 & 56.34 & 78.79 & 32.84 & 86.56 & 96.47 & 10.13 & 39.23 & 59.28 & 54.68 & 69.97 & 79.64 & 51.77 & 82.24 & 88.95 & \underline{50.85} \\ \hline
GE & 20.60 & 34.98 & 87.18 & 38.94 & 60.69 & 79.02 & 61.70 & 81.73 & 73.49 & 21.37 & 36.21 & 87.59 & 38.89 & 54.09 & 86.13 & 54.24 & 68.11 & 84.22 \\
SDCN & 1.56 & 24.23 & 86.65 & 0.37 & 33.74 & 76.78 & 1.66 & 40.64 & 70.92 & 1.61 & 24.00 & 87.10 & 0.88 & 25.29 & 85.09 & 1.25 & 26.27 & 83.45 \\
MCP & 0.56 & 24.12 & 84.20 & 11.66 & 42.72 & 65.46 & 63.34 & 84.18 & 19.60 & 0.29 & 23.86 & 85.03 & 0.00 & 25.16 & 83.21 & 0.00 & 25.90 & 81.91 \\
DMoN & 30.29 & 47.46 & 57.45 & 55.36 & 76.08 & 36.38 & 80.64 & 93.41 & 11.92 & 29.70 & 46.06 & 59.16 & 49.43 & 61.52 & 60.48 & 56.68 & 67.24 & 59.51 \\
S3GC & 40.69 & 59.12 & 53.40 & \underline{68.65} & 86.47 & 32.97 & \underline{89.14} & \underline{97.36} & \textbf{9.85} & 40.69 & 59.06 & 53.36 & 70.95 & 82.58 & 52.56 & 77.77 & 86.21 & 52.74 \\
MAGI & 40.77 & 57.49 & 59.65 & 65.65 & 84.20 & 33.19 & 88.51 & 97.20 & \textbf{9.85} & 40.99 & 58.47 & 59.43 & 67.14 & 81.04 & 63.93 & 76.31 & 87.14 & 72.46 \\ \hline
\textbf{AST} & \textbf{46.42} & \textbf{67.80} & \textbf{52.01} & \textbf{68.83} & \textbf{88.73} & \textbf{31.74} & \textbf{89.38} & \textbf{97.60} & 9.95 & \textbf{46.45} & \textbf{68.51} & \textbf{51.72} & \textbf{77.53} & \textbf{88.84} & \textbf{51.14} & \textbf{89.81} & \textbf{95.44} & \textbf{50.57} \\
Impv.\% & \textbf{5.12} & \textbf{3.15} & 0.57 & 0.26 & \textbf{2.47} & 0.16 & 0.27 & 0.25 & - & \textbf{3.75} & \textbf{2.88} & \textbf{2.18} & \textbf{6.76} & \textbf{3.80} & \textbf{1.08} & \textbf{3.37} & \textbf{1.63} & 0.55 \\ \hline
\end{tabular}
\end{table}

\textbf{Experiment Environments \& Implementations}.
All the experiments were conducted on a server with one Intel Xeon Gold 6430 CPU, one $24$GB memory GPU, $120$GB main memory, and Ubuntu Linux OS.
We implemented \textit{NJW}, \textit{SCORE}, \textit{RSC}, \textit{SCORE+}, \textit{ISC}, and ASCENT using \texttt{Python}, including the sparse ED for small datasets (e.g., \textbf{PPI}, \textbf{Wiki}, and \textbf{BlogCatalog}) and the Locally Optimal Block Preconditioned Conjugate Gradient (LOBPCG) solver \cite{knyazev2001toward} for large datasets (e.g., \textbf{obgn-Protein}, \textbf{RoadCA}, \textbf{LiveJournal}, \textbf{ogbn-Products}, and \textbf{Orkut}) supported by \texttt{SciPy}. We adopted the official open-source implementations of \textit{IBHCD}, \textit{CDBH}, and all the rest DGC baselines. Each DGC method was implemented via \texttt{PyTorch} or \texttt{TensorFlow} and thus was run on the GPU.
We have anonymously provided our code\footnote{https://anonymous.4open.science/r/ASCENT-714A} and will make it public if accepted.

\begin{table}[]\scriptsize
\caption{Evaluation results on real datasets, where the best and second-best metrics are in \textbf{bold} and \underline{underlined}; OOT and OOM represent out-of-time and out-of-memory exceptions.}\label{Tab:Eva-Res}
\centering
\begin{tabular}{p{0.75cm}|p{0.35cm}p{0.35cm}p{0.55cm}|p{0.35cm}p{0.35cm}p{0.55cm}|p{0.35cm}p{0.35cm}p{0.55cm}|p{0.35cm}p{0.35cm}p{0.55cm}|p{0.35cm}p{0.35cm}p{0.55cm}|p{0.35cm}p{0.35cm}p{0.55cm}}
\hline
\multirow{3}{*}{} & \multicolumn{3}{c|}{\textbf{Caltech}} & \multicolumn{3}{c|}{\textbf{Simmons}} & \multicolumn{3}{c|}{\textbf{PolBlogs}} & \multicolumn{3}{c|}{\textbf{BioGrid}} & \multicolumn{3}{c|}{\textbf{PPI}} & \multicolumn{3}{c}{\textbf{Wiki}} \\ \cline{2-19} 
 & NMI & AC & Cond & NMI & AC & Cond & NMI & AC & Cond & NMI & AC & Cond & \multicolumn{3}{c|}{Cond($\downarrow$,\%)} & \multicolumn{3}{c}{Cond($\downarrow$,\%)} \\ \cline{2-19} 
 & ($\uparrow$,\%) & ($\uparrow$,\%) & ($\downarrow$,\%) & ($\uparrow$,\%) & ($\uparrow$,\%) & ($\downarrow$,\%) & ($\uparrow$,\%) & ($\uparrow$,\%) & ($\downarrow$,\%) & ($\uparrow$,\%) & ($\uparrow$,\%) & ($\downarrow$,\%) & $K$=2 & 8 & 32 & $K$=2 & 8 & 32 \\ \hline
NJW & 62.13 & 75.39 & 50.76 & 67.96 & 73.44 & 33.87 & 0.06 & 51.88 & 26.94 & 41.83 & 12.84 & 66.20 & 31.09 & 61.18 & 72.06 & 40.16 & 74.32 & \underline{85.32} \\
SCORE & 56.39 & 69.05 & 50.12 & 58.53 & 76.39 & 29.92 & 72.50 & 95.25 & 7.67 & 13.93 & 7.37 & 90.90 & 28.82 & 69.07 & 91.81 & 40.29 & 82.73 & 92.08 \\
RSC & 58.58 & 71.05 & 49.86 & 61.52 & 78.61 & 28.88 & 71.33 & 94.76 & 7.34 & \textbf{43.64} & \underline{13.52} & \underline{66.07} & \underline{26.02} & 55.12 & 68.80 & 37.75 & \underline{73.06} & \underline{85.32} \\
SC+ & 69.14 & 82.85 & 48.44 & 72.95 & 88.81 & 27.41 & \underline{73.08} & \underline{95.33} & 7.53 & 24.36 & 9.44 & 82.02 & 26.71 & 55.30 & 79.79 & 39.06 & 75.69 & 86.87 \\
ISC & \underline{70.28} & \underline{83.73} & \underline{48.32} & \underline{73.57} & \underline{89.36} & \underline{27.35} & 72.67 & 95.09 & 7.35 & 43.21 & 13.26 & \underline{66.07} & 26.22 & 54.01 & 68.17 & 37.91 & 73.31 & 85.78 \\
IBHCD & 59.18 & 70.20 & 49.51 & 64.77 & 81.71 & 28.33 & \underline{73.30} & \textbf{95.34} & \underline{7.33} & \multicolumn{3}{c|}{OOT} & 27.62 & 56.18 & \underline{66.33} & \underline{37.58} & 73.47 & 89.43 \\
CDBH & 61.53 & 72.98 & 49.72 & 64.77 & 81.71 & 28.33 & \underline{73.30} & \textbf{95.34} & \underline{7.33} & \multicolumn{3}{c|}{OOT} & 26.72 & 56.18 & 69.61 & \textbf{37.43} & 74.97 & 86.14 \\ \hline
GE & 36.75 & 44.44 & 68.97 & 49.19 & 58.86 & 48.13 & 6.15 & 51.88 & 50.13 & 31.74 & 11.28 & 98.51 & 43.11 & 84.96 & 96.15 & 52.08 & 89.02 & 97.23 \\
SDCN & 28.50 & 34.03 & 74.79 & 38.28 & 53.91 & 51.39 & 14.96 & 62.93 & 28.48 & 20.77 & 8.44 & 95.83 & 50.00 & 86.01 & 94.69 & 50.00 & 87.30 & 96.41 \\
MCP & 50.57 & 63.32 & 58.16 & 64.66 & 82.90 & 29.80 & 58.15 & 86.56 & 16.26 & 0.00 & 4.43 & 98.77 & 50.00 & 87.50 & 96.88 & 50.00 & 87.50 & 96.88 \\
DMoN & 66.29 & 72.47 & 53.97 & 63.64 & 81.20 & 28.00 & 71.16 & 94.91 & 7.47 & 41.73 & 12.91 & 79.74 & 26.31 & 55.43 & 70.13 & 37.09 & 77.25 & 91.59 \\
S3GC & 61.23 & 71.69 & 51.89 & 66.62 & 82.02 & 28.44 & 71.53 & 95.04 & 7.38 & 33.20 & 11.66 & 73.36 & 29.88 & 67.74 & 72.19 & 51.76 & 84.74 & 89.05 \\
MAGI & 66.26 & 78.17 & 49.16 & 71.50 & 86.77 & 27.73 & 68.87 & 94.39 & 7.38 & 10.68 & 6.88 & 84.65 & 28.45 & \textbf{51.75} & 74.96 & 42.99 & 81.02 & 92.78 \\ \hline
\textbf{AST} & \textbf{71.20} & \textbf{84.41} & \textbf{48.28} & \textbf{74.06} & \textbf{89.62} & \textbf{27.34} & \textbf{73.48} & \textbf{95.34} & \textbf{7.31} & \underline{43.28} & \textbf{13.59} & \textbf{64.88} & \textbf{25.77} & \underline{52.73} & \textbf{65.54} & \textbf{37.43} & \textbf{72.00} & \textbf{84.77} \\
Impv.\% & \textbf{1.31} & 0.81 & 0.08 & 0.67 & 0.29 & 0.04 & 0.25 & - & 0.27 & - & 0.52 & \textbf{1.80} & 0.96 & - & \textbf{1.19} & - & \textbf{1.45} & 0.64 \\ \hline
\multirow{3}{*}{} & \multicolumn{3}{c|}{\textbf{BlogCatalog}} & \multicolumn{3}{c|}{\textbf{ogbn-Protein}} & \multicolumn{3}{c|}{\textbf{RoadCA}} & \multicolumn{3}{c|}{\textbf{LiveJournal}} & \multicolumn{3}{c|}{\textbf{ogbn-Products}} & \multicolumn{3}{c}{\textbf{Orkut}} \\ \cline{2-19} 
 & \multicolumn{3}{c|}{Cond($\downarrow$,\%)} & \multicolumn{3}{c|}{Cond($\downarrow$,\%)} & \multicolumn{3}{c|}{Cond($\downarrow$,\%)} & \multicolumn{3}{c|}{Cond($\downarrow$,\%)} & \multicolumn{3}{c|}{Cond($\downarrow$,\%)} & \multicolumn{3}{c}{Cond($\downarrow$,\%)} \\ \cline{2-19} 
 & $K$=2 & 8 & 32 & $K$=2 & 8 & 32 & $K$=2 & 8 & 32 & $K$=2 & 8 & 32 & $K$=2 & 8 & 32 & $K$=2 & 8 & 32 \\ \hline
NJW & 29.35 & 67.83 & 82.79 & \underline{6.34} & \underline{11.92} & 41.36 & \textbf{4.50} & \underline{11.09} & \underline{15.08} & \underline{5.77} & \underline{20.21} & \underline{36.92} & 6.03 & 8.67 & 10.31 & \underline{0.64} & 21.50 & 33.88 \\
SCORE & 29.65 & 77.84 & 93.67 & 22.38 & 44.69 & 82.42 & 50.00 & 87.50 & 96.88 & 50.00 & 87.50 & 96.88 & 50.00 & 87.50 & 96.88 & 50.00 & 87.50 & 96.88 \\
RSC & 29.24 & 66.56 & 81.74 & 12.03 & 15.87 & 36.63 & 12.31 & 23.13 & 29.57 & 23.34 & 50.82 & 60.82 & 17.55 & 34.36 & 26.54 & 22.84 & 30.68 & 42.15 \\
SC+ & 29.33 & 69.95 & 86.19 & 7.09 & 21.92 & 64.33 & 50.00 & 87.50 & 96.85 & 49.96 & 87.48 & 96.87 & 49.32 & 86.30 & 96.50 & 49.98 & 87.49 & 96.87 \\
ISC & \underline{29.26} & \underline{65.02} & \underline{81.43} & 6.93 & 14.88 & \underline{34.64} & 5.84 & 14.24 & 18.10 & \textbf{5.74} & 21.93 & 36.73 & \underline{5.29} & \underline{8.02} & \underline{10.18} & \textbf{0.61} & \underline{16.27} & \underline{31.01} \\
IBHCD & 29.31 & OOT & OOT & \multicolumn{3}{c|}{OOT} & \multicolumn{3}{c|}{OOT} & \multicolumn{3}{c|}{OOT} & \multicolumn{3}{c|}{OOT} & \multicolumn{3}{c}{OOT} \\
CDBH & 29.31 & 68.17 & 82.03 & \multicolumn{3}{c|}{OOT} & \multicolumn{3}{c|}{OOT} & \multicolumn{3}{c|}{OOT} & \multicolumn{3}{c|}{OOT} & \multicolumn{3}{c}{OOT} \\ \hline
GE & 49.59 & 87.49 & 96.78 & \multicolumn{3}{c|}{OOM} & \multicolumn{3}{c|}{OOM} & \multicolumn{3}{c|}{OOM} & \multicolumn{3}{c|}{OOM} & \multicolumn{3}{c}{OOM} \\
SDCN & 49.87 & 85.44 & 95.80 & \multicolumn{3}{c|}{OOM} & \multicolumn{3}{c|}{OOM} & \multicolumn{3}{c|}{OOM} & \multicolumn{3}{c|}{OOM} & \multicolumn{3}{c}{OOM} \\
MCP & 34.15 & 86.97 & 96.82 & \multicolumn{3}{c|}{OOM} & \multicolumn{3}{c|}{OOM} & \multicolumn{3}{c|}{OOM} & \multicolumn{3}{c|}{OOM} & \multicolumn{3}{c}{OOM} \\
DMoN & 42.89 & 75.95 & 90.28 & \multicolumn{3}{c|}{OOM} & \multicolumn{3}{c|}{OOM} & \multicolumn{3}{c|}{OOM} & \multicolumn{3}{c|}{OOM} & \multicolumn{3}{c}{OOM} \\
S3GC & 38.28 & 72.82 & 87.43 & 6.50 & 13.46 & 51.60 & 8.4 & 26.75 & 48.54 & 16.36 & 58.31 & 85.70 & 13.46 & 41.24 & 67.72 & 2.55 & 45.19 & 80.18 \\
MAGI & 30.70 & 75.64 & 92.05 & 15.38 & 33.81 & 63.74 & 14.27 & 33.65 & 48.94 & 23.75 & 64.96 & 85.52 & 19.68 & 48.01 & 67.50 & 15.20 & 66.37 & 88.51 \\ \hline
\textbf{AST} & \textbf{29.23} & \textbf{64.22} & \textbf{80.68} & \textbf{3.47} & \textbf{11.06} & \textbf{33.20} & \underline{4.53} & \textbf{10.89} & \textbf{14.67} & 6.34 & \textbf{19.90} & \textbf{35.61} & \textbf{4.80} & \textbf{7.97} & \textbf{9.94} & \textbf{0.61} & \textbf{15.13} & \textbf{29.96} \\
Impv.\% & 0.10 & \textbf{1.23} & 0.92 & \textbf{45.27} & \textbf{7.21} & \textbf{4.16} & - & \textbf{1.80} & \textbf{2.72} & - & \textbf{1.53} & \textbf{3.55} & \textbf{9.26} & 0.62 & \textbf{2.36} & - & \textbf{7.01} & \textbf{3.39} \\ \hline
\end{tabular}
\end{table}

\subsection{Quality Evaluation}

For each setting of LFR-net, we independently generated $100$ graphs and recorded the mean of all quality metrics over them. On each real dataset, we repeated the evaluation procedure over $5$ random seeds and recorded the mean of each metric.
Average results on synthetic and real graphs are depicted in Tables~\ref{Tab:Eva-Res-LFR} and \ref{Tab:Eva-Res}, where a metric is in \textbf{bold} or \underline{underlined} if it performs the best or second-best; OOM and OOT denote the out-of-memory and out-of-time exceptions.

In most cases, SC methods have much better quality than DGC baselines. It indicates that some GNN-based approaches, with standard strategies to extract auxiliary features from topology, may fail to handle the high degree heterogeneity and weak clustering structures, although some of them are claimed to be effective for attributed graphs. Some DGC approaches also encounter the out-of-memory exception on large-scale graphs (e.g., \textbf{Orkut}), due to the reconstruction of an $N \times N$ matrix (e.g., normalized adjacency matrices in \textit{GraphEncoder}) or space-consuming training procedures. In contrast, SC methods can derive feasible clustering results on all the datasets.

Compared with SC algorithms using the standard procedures about graph Laplacian (i.e., ASCENT and those summarized in Table~\ref{Tab:Alg-Sum}), \textit{IBHCD} and \textit{CDBH} that follow a new SC framework may easily suffer from low efficiency and cannot even derive feasible clustering results within $10^4$ seconds (i.e., out-of-time exception), especially for large $N$s and $K$s.

ASCENT performs the best in most cases and can achieve significantly better quality than other DCSC baselines (e.g., with improvements from about $1\%$ to $40\%$) in some cases.
It validates that ASCENT, a simple yet effective extension of DCSC, is more powerful in handling the high degree heterogeneity and weak clustering structures.

\begin{table}[]\scriptsize
\caption{Evaluation of runtime$\downarrow$ (sec) on real datasets.}\label{Tab:Eva-Time}
\centering
\begin{tabular}{p{0.95cm}|p{0.5cm}p{0.25cm}p{0.5cm}p{0.55cm}|p{0.5cm}p{0.25cm}p{0.5cm}p{0.55cm}|p{0.5cm}p{0.25cm}p{0.5cm}p{0.55cm}|p{0.5cm}p{0.25cm}p{0.5cm}p{0.5cm}}
\hline
\multirow{2}{*}{} & \multicolumn{4}{c|}{\textbf{Caltech}} & \multicolumn{4}{c|}{\textbf{Simmons}} & \multicolumn{4}{c|}{\textbf{PolBlogs}} & \multicolumn{4}{c}{\textbf{BioGrid}} \\ \cline{2-17} 
 & Total & $\tau$ & ED/T & C/I & Total & $\tau$ & ED/T & C/I & Total & $\tau$ & ED/T & C/I & Total & $\tau$ & ED/T & C/I \\ \hline
NJW & 0.09 & N/A & 0.09 & 0.09 & 0.15 & N/A & 0.14 & 0.01 & 0.10 & N/A & 0.10 & 0.01 & 0.90 & N/A & 0.83 & 0.06 \\
SCORE & 0.07 & N/A & 0.06 & 0.01 & 0.11 & N/A & 0.09 & 0.01 & 0.08 & N/A & 0.07 & 0.01 & 0.54 & N/A & 0.42 & 0.12 \\
RSC & 0.10 & N/A & 0.09 & 0.01 & 0.15 & N/A & 0.14 & 0.01 & 0.12 & N/A & 0.11 & 0.01 & 0.75 & N/A & 0.63 & 0.12 \\
SCORE+ & 0.15 & N/A & 0.14 & 0.01 & 0.25 & N/A & 0.24 & 0.01 & 0.22 & N/A & 0.21 & 0.01 & 0.98 & N/A & 0.85 & 0.13 \\
ISC & 0.10 & N/A & 0.09 & 0.01 & 0.16 & N/A & 0.15 & 0.01 & 0.13 & N/A & 0.12 & 0.01 & 0.76 & N/A & 0.64 & 0.12 \\ 
IBHCD & 0.87 & N/A & - & - & 0.91 & N/A & - & - & 0.73 & N/A & - & - & \multicolumn{4}{c}{OOT} \\
CDBH & 3.09 & N/A & - & - & 3.65 & N/A & - & - & 1.11 & N/A & - & - & \multicolumn{4}{c}{OOT} \\ \hline
GE & 1.37 & N/A & 1.36 & 0.01 & 2.19 & N/A & 2.18 & 0.01 & 1.54 & N/A & 1.53 & 0.01 & 12.84 & N/A & 12.38 & 0.46 \\
SDCN & 0.29 & N/A & 0.29 & 0.00 & 0.38 & N/A & 0.36 & 0.02 & 11.32 & N/A & 11.29 & 0.03 & 38.12 & N/A & 38.08 & 0.04 \\
MCP & 25.18 & N/A & 25.17 & 0.02 & 40.61 & N/A & 40.60 & 0.02 & 9.45 & N/A & 9.44 & 0.01 & 11.59 & N/A & 11.57 & 0.02 \\
DMoN & 16.50 & N/A & 16.49 & 0.01 & 19.68 & N/A & 19.66 & 0.02 & 17.76 & N/A & 17.74 & 0.01 & 28.79 & N/A & 28.76 & 0.04 \\
S3GC & 22.24 & N/A & 22.22 & 0.01 & 25.34 & N/A & 25.32 & 0.02 & 27.93 & N/A & 27.91 & 0.01 & 33.78 & N/A & 33.45 & 0.33 \\
MAGI & 4.04 & N/A & 4.02 & 0.02 & 2.72 & N/A & 2.71 & 0.01 & 2.70 & N/A & 2.69 & 0.01 & 22.41 & N/A & 22.26 & 0.15 \\ \hline
\textbf{ASCENT} & 0.10 & \text{2e-4} & 0.08 & 0.01 & 0.16 & \text{3e-4} & 0.15 & 0.01 & 0.12 & \text{2e-4} & 0.10 & 0.01 & 0.72 & \text{9e-4} & 0.59 & 0.13 \\ \hline
\multirow{2}{*}{} & \multicolumn{4}{c|}{\textbf{PPI}} & \multicolumn{4}{c|}{\textbf{Wiki}} & \multicolumn{4}{c|}{\textbf{BlogCatalog}} & \multicolumn{4}{c}{\textbf{ogbn-Protein}} \\ \cline{2-17} 
 & Total & $\tau$ & ED/T & C/I & Total & $\tau$ & ED/T & C/I & Total & $\tau$ & ED/T & C/I & Total & $\tau$ & ED/T & C/I \\ \hline
NJW & 0.37 & N/A & 0.34 & 0.03 & 0.86 & N/A & 0.83 & 0.03 & 1.64 & N/A & 1.60 & 0.04 & 196.10 & N/A & 195.72 & 0.38 \\
SCORE & 0.43 & N/A & 0.39 & 0.04 & 0.46 & N/A & 0.42 & 0.04 & 1.15 & N/A & 1.09 & 0.06 & 96.49 & N/A & 96.24 & 0.25 \\
RSC & 0.46 & N/A & 0.42 & 0.03 & 0.75 & N/A & 0.70 & 0.05 & 1.70 & N/A & 1.62 & 0.09 & 148.49 & N/A & 148.13 & 0.36 \\
SCORE+ & 0.58 & N/A & 0.53 & 0.05 & 0.78 & N/A & 0.71 & 0.08 & 2.04 & N/A & 1.94 & 0.10 & 174.61 & N/A & 174.29 & 0.32 \\
ISC & 0.60 & N/A & 0.55 & 0.05 & 0.82 & N/A & 0.74 & 0.08 & 1.96 & N/A & 1.87 & 0.08 & 174.39 & N/A & 174.04 & 0.35 \\ 
IBHCD & 612.41 & N/A & - & - & 803.66 & N/A & - & - & \multicolumn{4}{c|}{OOT} & \multicolumn{4}{c}{OOT} \\
CDBH & 447.01 & N/A & - & - & 345.07 & N/A & - & - & 4825.14 & N/A & - & - & \multicolumn{4}{c}{OOT} \\ \hline
GE & 0.74 & N/A & 0.70 & 0.04 & 18.98 & N/A & 18.83 & 0.15 & 125.91 & N/A & 125.64 & 0.27 & \multicolumn{4}{c}{OOM} \\
SDCN & 34.37 & N/A & 34.33 & 0.04 & 44.85 & N/A & 44.81 & 0.04 & 58.81 & N/A & 58.61 & 0.20 & \multicolumn{4}{c}{OOM} \\
MCP & 7.94 & N/A & 7.91 & 0.02 & 15.09 & N/A & 15.06 & 0.03 & 25.08 & N/A & 25.01 & 0.07 & \multicolumn{4}{c}{OOM} \\
DMoN & 214.45 & N/A & 214.42 & 0.03 & 65.10 & N/A & 65.05 & 0.05 & 90.03 & N/A & 89.86 & 0.17 & \multicolumn{4}{c}{OOM} \\
S3GC & 26.62 & N/A & 26.41 & 0.21 & 19.09 & N/A & 18.95 & 0.14 & 40.65 & N/A & 40.37 & 0.29 & 228.45 & N/A & 248.42 & 2.87 \\
MAGI & 11.98 & N/A & 11.92 & 0.06 & 19.67 & N/A & 19.57 & 0.10 & 19.42 & N/A & 19.32 & 0.11 & 2496.72 & N/A & 2488.68 & 8.04 \\ \hline
\textbf{ASCENT} & 0.51 & \text{8e-4} & 0.46 & 0.05 & 0.68 & \text{9e-4} & 0.62 & 0.06 & 1.60 & 0.004 & 1.51 & 0.09 & 163.46 & 0.20 & 162.95 & 0.30 \\ \hline
\multirow{2}{*}{} & \multicolumn{4}{c|}{\textbf{RoadCA}} & \multicolumn{4}{c|}{\textbf{LiveJournal}} & \multicolumn{4}{c|}{\textbf{ogbn-Products}} & \multicolumn{4}{c}{\textbf{Orkut}} \\ \cline{2-17} 
 & Total & $\tau$ & ED/T & C/I & Total & $\tau$ & ED/T & C/I & Total & $\tau$ & ED/T & C/I & Total & $\tau$ & ED/T & C/I \\ \hline
NJW & 151.52 & N/A & 139.41 & 12.11 & 386.50 & N/A & 361.58 & 24.92 & 566.75 & N/A & 557.17 & 9.57 & 777.96 & N/A & 768.46 & 9.50 \\
SCORE & 210.59 & N/A & 205.58 & 5.01 & 426.21 & N/A & 414.23 & 11.98 & 481.20 & N/A & 473.76 & 7.44 & 649.19 & N/A & 640.29 & 8.91 \\
RSC & 184.60 & N/A & 171.99 & 12.60 & 374.68 & N/A & 346.00 & 28.68 & 455.25 & N/A & 445.34 & 9.91 & 606.13 & N/A & 595.72 & 10.41 \\
SCORE+ & 177.37 & N/A & 172.69 & 4.68 & 147.57 & N/A & 135.42 & 12.15 & 239.05 & N/A & 231.63 & 7.42 & 421.95 & N/A & 413.77 & 8.18 \\
ISC & 182.61 & N/A & 169.08 & 13.53 & 400.84 & N/A & 369.43 & 31.41 & 520.69 & N/A & 513.68 & 7.01 & 717.78 & N/A & 707.77 & 10.01 \\ 
IBHCD & \multicolumn{4}{c|}{OOT} & \multicolumn{4}{c|}{OOT} & \multicolumn{4}{c|}{OOT} & \multicolumn{4}{c}{OOT} \\
CDBH & \multicolumn{4}{c|}{OOT} & \multicolumn{4}{c|}{OOT} & \multicolumn{4}{c|}{OOT} & \multicolumn{4}{c}{OOT} \\ \hline
GE & \multicolumn{4}{c|}{OOM} & \multicolumn{4}{c|}{OOM} & \multicolumn{4}{c|}{OOM} & \multicolumn{4}{c}{OOM} \\
SDCN & \multicolumn{4}{c|}{OOM} & \multicolumn{4}{c|}{OOM} & \multicolumn{4}{c|}{OOM} & \multicolumn{4}{c}{OOM} \\
MCP & \multicolumn{4}{c|}{OOM} & \multicolumn{4}{c|}{OOM} & \multicolumn{4}{c|}{OOM} & \multicolumn{4}{c}{OOM} \\
DMoN & \multicolumn{4}{c|}{OOM} & \multicolumn{4}{c|}{OOM} & \multicolumn{4}{c|}{OOM} & \multicolumn{4}{c}{OOM} \\
S3GC & 194.39 & N/A & 154.81 & 44.30 & 512.44 & N/A & 482.12 & 45.20 & 595.90 & N/A & 605.01 & 73.33 & 798.82 & N/A & 828.59 & 29.58 \\
MAGI & 7159.19 & N/A & 7139.72 & 19.48 & 7998.09 & N/A & 7872.18 & 125.91 & 7273.09 & N/A & 7142.89 & 130.20 & 7568.24 & N/A & 7463.25 & 104.99 \\ \hline
\textbf{ASCENT} & 133.46 & 0.99 & 119.59 & 12.88 & 387.97 & 12.41 & 348.03 & 27.54 & 567.44 & 39.36 & 520.13 & 7.95 & 740.52 & 44.15 & 684.90 & 11.47 \\ \hline
\end{tabular}
\end{table}

\subsection{Efficiency Analysis}

We further evaluated the efficiency of each method in terms of its overall runtime (sec) to get a feasible clustering result.
In particular, we recorded the time of different steps for each method.
Results of the efficiency analysis on real datasets are depicted in Table~\ref{Tab:Eva-Time}, where (\romannumeral1) `$\tau$', (\romannumeral2) `ED', (\romannumeral3) `C', (\romannumeral4) `T', and (\romannumeral5) `I' denote the time of (\romannumeral1) deriving node-wise corrections $\{ \tau_i \}$ (only for ASCENT), (\romannumeral2) ED of graph Laplacian, (\romannumeral3) $K$Means clustering, (\romannumeral4) training of a DGC method, and (\romannumeral5) inference of a DGC baseline.

Compared with DGC approaches (e.g., \textit{GraphEncoder}, \textit{S3GC}, and \textit{MAGI}), which involve a time-consuming learning procedure (e.g., gradient descent to iteratively update model parameters), all the SC methods except \textit{IBHCD} and \textit{CDBH} can achieve significantly better efficiency in most cases.
In particular, the overall runtime of \textit{IBHCD} and \textit{CDBH} may even be longer than those of DGC baselines, indicting their limitations of low efficiency.
ED is the major bottleneck of most SC algorithms.

For ASCENT, the derivation of node-wise corrections $\{ \tau_i \}$ would not significantly increase the overall runtime compared with other SC baselines. Surprisingly, ASCENT can even achieve better overall efficiency than some DCSC baselines (e.g., \textit{ISC} on \textbf{ogbn-Protein}, \textbf{RoadCA}, and \textbf{LiveJournal}), due to its faster ED. In summary, ASCENT can still ensure high efficiency close to that of conventional SC methods and much better than that of DGC approaches.

\begin{figure}
\centering
 \begin{minipage}{0.24\linewidth}
 \subfigure[\textbf{NMI}$\uparrow$]{
  \includegraphics[width=\textwidth,trim=25 15 22 22,clip]{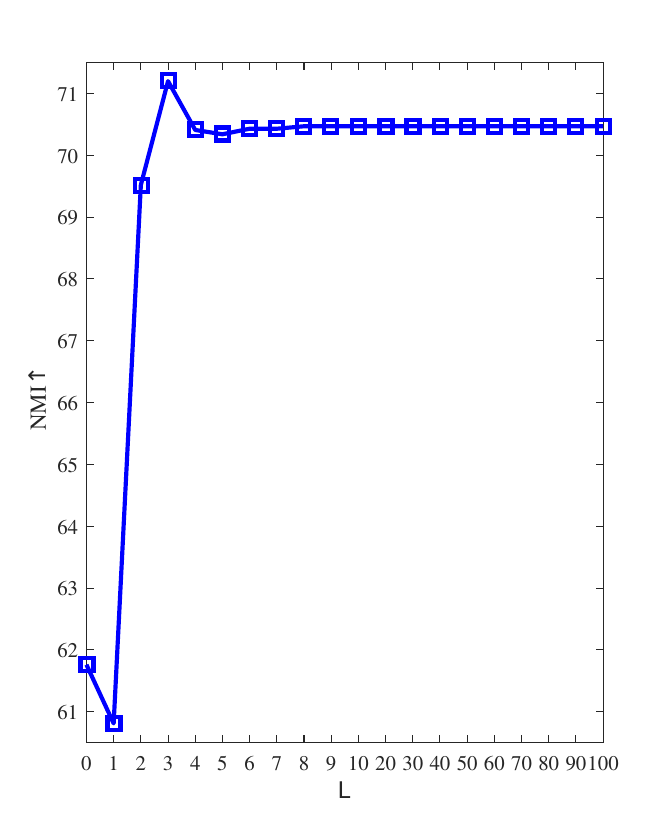}}
 \end{minipage}
 \begin{minipage}{0.24\linewidth}
 \subfigure[\textbf{AC}$\uparrow$]{
  \includegraphics[width=\textwidth,trim=25 15 22 22,clip]{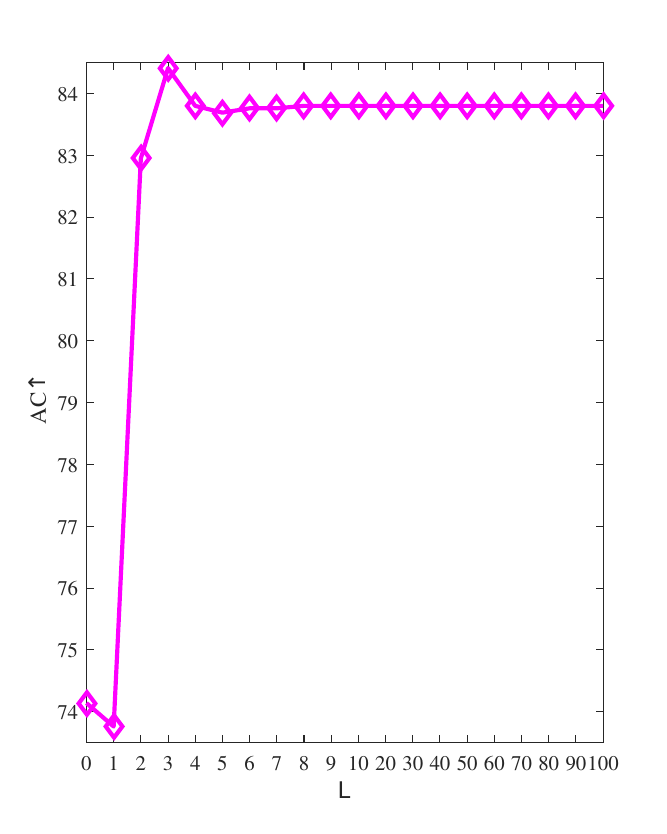}}
 \end{minipage}
 \begin{minipage}{0.24\linewidth}
 \subfigure[\textbf{Conductance}$\downarrow$]{
  \includegraphics[width=\textwidth,trim=23 15 22 22,clip]{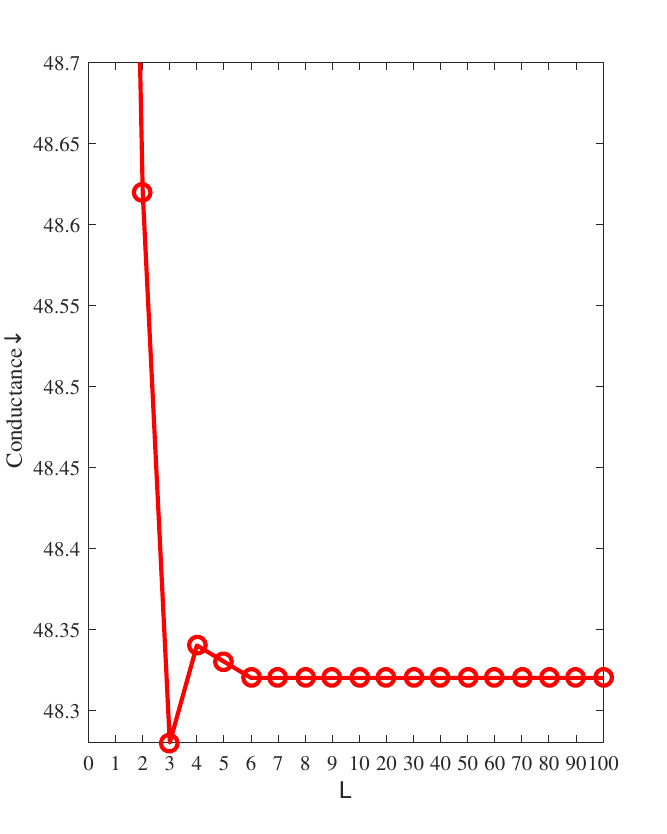}}
 \end{minipage}
\caption{Parameter analysis of $L$ on \textbf{Caltech}.}\label{Fig:Param-L}
\end{figure}

\begin{figure}
\centering
 \begin{minipage}{0.325\linewidth}
 \subfigure[\textbf{NMI}$\uparrow$]{
  \includegraphics[width=\textwidth,trim=58 10 46 5,clip]{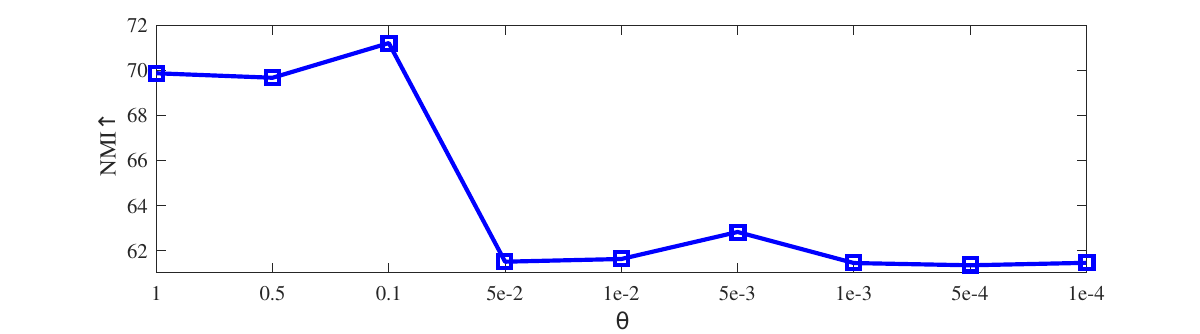}}
 \end{minipage}
 \begin{minipage}{0.325\linewidth}
 \subfigure[\textbf{AC}$\uparrow$]{
  \includegraphics[width=\textwidth,trim=58 10 46 10,clip]{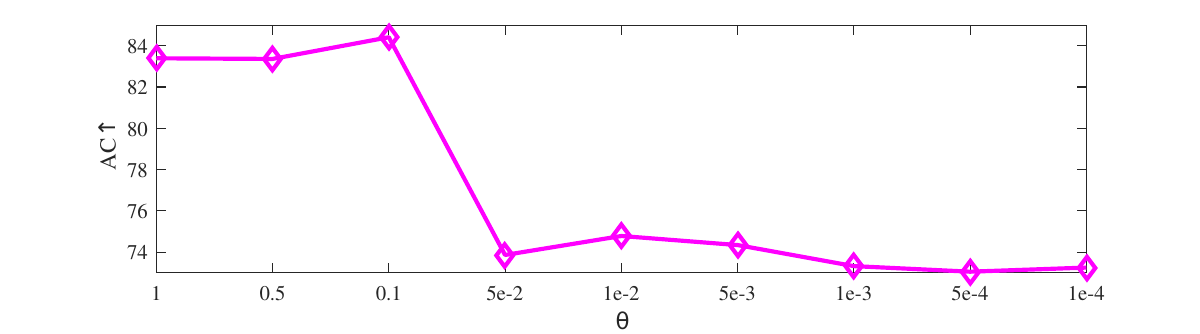}}
 \end{minipage}
 \begin{minipage}{0.325\linewidth}
 \subfigure[\textbf{Conductance}$\downarrow$]{
  \includegraphics[width=\textwidth,trim=58 10 46 10,clip]{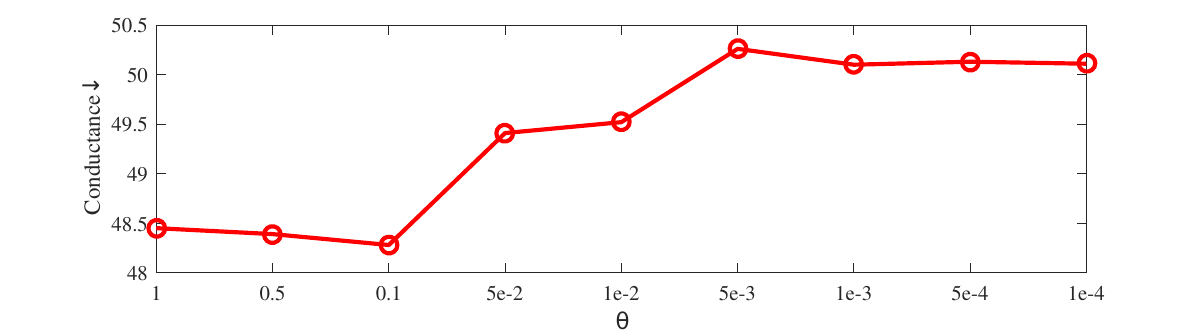}}
 \end{minipage}
\caption{Parameter analysis of $\theta$ on \textbf{Caltech}.}\label{Fig:Param-theta}
\end{figure}

\subsection{Parameter Analysis}
We further tested the effect of $L$ for ASCENT. Example results on \textbf{Caltech} are visualized in Fig.~\ref{Fig:Param-L}, where we adjusted $L \in \{ 0, 1, \cdots, 10, 20, \cdots, 100\}$.
When $L = 0$ or is small (e.g, $L \le 2$ for \textbf{Caltech}), ASCENT suffers from poor clustering quality, which can be improved as $L$ increases. Consistent with our case study in Fig.~\ref{Fig:Case}, the clustering quality gradually converges due to over-smoothing, with the increase of $L$. In particular, ASCENT achieves the best quality with a small $L$ (i.e., $L=3$ for \textbf{Caltech}) before it reduces to conventional DCSC methods with a constant $\tau$.

Parameter analysis results of $\theta$ on \textbf{Caltech} are visualized in Fig.~\ref{Fig:Param-theta}, where we set $\theta \in \{ 1, 0.5, 0.1, 5\times10^{-2}, 10^{-2}, 5\times10^{-3}, 10^{-3}, 5\times10^{-4}, 10^{-4}\}$.
According to Fig.~\ref{Fig:Param-theta}, the quality of ASCENT may be sensitive to the setting of $\theta$.
We determined parameter settings of $\{ \theta, L \}$ for ASCENT on all the datasets based on the unsupervised average conductance metric. Table~\ref{Tab:Eva-Param} summarizes the suggested settings.

\begin{table}[]\scriptsize
\caption{Recommended parameter settings of ASCENT.}\label{Tab:Eva-Param}
\centering
\begin{tabular}{l|llllllll}
\hline
 & \textbf{LFR-1} & \textbf{LFR-2} & \textbf{SBM-1} & \multicolumn{1}{l|}{\textbf{SBM-2}} & \textbf{Caltech} & \textbf{Simmons} & \textbf{PolBlogs} & \textbf{BioGrid} \\ \hline
$\theta$ & 1.0 & 1.0 & 0.01 & \multicolumn{1}{l|}{0.01} & 0.1 & 0.2 & 0.05 & 0.1 \\
$L$ & 4 & 2 & 1 & \multicolumn{1}{l|}{1} & 3 & 3 & 1 & 5 \\ \hline
 & \textbf{PPI} & \textbf{Wiki} & \multicolumn{1}{l|}{\textbf{BlogCatalog}} & \textbf{ogbn-Protein} & \textbf{RoadCA} & \textbf{LiveJournal} & \textbf{ogbn-Products} & \textbf{Orkut} \\ \hline
$\theta$ & 0.1 & 0.1 & \multicolumn{1}{l|}{0.1} & 0.1 & 0.1 & 1e-3 & 1e-3 & 0.01 \\
$L$ & 5 & 4 & \multicolumn{1}{l|}{7} & 2 & 50 & 50 & 60 & 50 \\ \hline
\end{tabular}
\end{table}

\section{Conclusion}\label{Sec:Con}
In this paper, we presented a condition-free analysis of DCSC from a pure spectral view. Different from prior studies based on random graph models (e.g., DCSBM), our analysis gives an upper bound for the number of mis-clustered nodes w.r.t. the optimal solution to average conductance minimization, without relying on additional conditions and random graph models. Our analysis also includes quantities that reveal impacts of (\romannumeral1) degree heterogeneity and (\romannumeral2) weakness of clustering structures to the clustering quality of DCSC.
Inspired by GNNs and their over-smoothing effect, we proposed ASCENT, a simple yet effective extension of DCSC. It follows a node-wise correction scheme that can assign nodes $\{ v_i \}$ with different corrections $\{ \tau_i \}$ via the GNN mean aggregation. In particular, ASCENT reduces to conventional DCSC methods when encountering over-smoothing. Some early stages before over-smoothing can potentially result in better clustering quality for ASCENT.
We conclude this paper by discussing limitations of this paper and future research directions.

\textbf{Clustering on Attributed Graphs}. As described in Section~\ref{Sec:Prob}, we followed the conventional problem statement of SC, where topology is the only available information source (i.e., without attributes), due to the complicated corrections between graph topology and attributes. In our future work, we will analyze DCSC on attributed graphs with the consideration of the possible inconsistency between the two sources \cite{newman2016structure,qin2018adaptive,wang2020gcn,qin2021dual}.

\textbf{Learnable Node-wise Corrections $\{ \tau_i \}$}. In ASCENT, we still set the node-wise corrections $\{ \tau_i \}$ by manually adjusting hyper-parameters $\{ \theta, L\}$. We plan to extend it to a more advanced setting with learnable node-wise corrections $\{ \tau_i \}$ and provide theoretical analysis combined with recent advances in spectral-based GNNs \cite{bo2021beyond,dong2021adagnn,qin2025efficient,qin2025infraredgp,chen2026spectral}.

\textbf{Other Graph Clustering Objectives}. In this study, we only considered the average conductance minimization objective (or equivalently normalized cut minimization) as defined in \textbf{Definition~\ref{Def:Cond-Min}}. SC can be considered as an approximated algorithm for a relaxed version of this objective. Some graph clustering algorithms may consider other objectives (e.g., ratio-cut minimization \cite{von2007tutorial} and modularity maximization \cite{newman2006modularity,qin2024pre,qin2024towards}) that have relations close to average conductance minimization. We also plan to further extend our analysis to these objectives.

\appendix
\section{Proof of Theorem~\ref{Th:Struc}}\label{App:Th-Struc}

Recall that we have ${\bf{ F}} := [\lambda_1 {\bf{u}}_1, \cdots, \lambda_{K+1} {\bf{u}}_{K+1}]$ as the rearranged spectral embedding of \textit{ISC}; ${\bf{G}} \in \mathbb{R}^{N \times K}$ encodes membership of the optimal partition $({\hat S}_1, \cdots, {\hat S}_K)$, where ${\bf{G}}_{ir} = \sqrt{d_i / \mu({\hat S}_r)}$ if $v_i \in {\hat S}_r$ and ${\bf{G}}_{ir} = 0$ otherwise.
Let ${\bf{U}} := [{\bf{u}}_1, \cdots, {\bf{u}}_{K+1}]$ be the arrangement of eigenvectors without reweighting and then ${\bf{F}} = {\bf{U}} {\bf{\Lambda}}$, where ${\bf{\Lambda }}: = {\mathop{\rm diag}\nolimits} ({\lambda _1}, \cdots ,{\lambda _{K + 1}})$ is a diagonal matrix w.r.t the leading $(K+1)$ eigenvalues.

For each $r \in [K]$, it is obvious that $|{{\bf{G}}_{:,r}}{|_2} = 1$. In particular, one can derive ${\bf{G}}_{:,r}$ via the linear combination of eigenvectors $\{ {\bf{u}}_1, \cdots, {\bf{u}}_N \}$ w.r.t. eigenvalues $\{ \lambda_1, \cdots, \lambda_N \}$ of regularized graph Laplacian ${\bf{L}}_{\tau}$. Namely, we have ${{\bf{G}}_{:,r}} = \sum\nolimits_{i = 1}^N {{h_{ir}}{\lambda _i}{{\bf{u}}_i}}$ with $\{ h_{ir} \}$ as corresponding weights for the linear combination. Let ${{{\bf{\hat u}}}_r}: = \sum\nolimits_{i = 1}^{K + 1} {{h_{ir}}{\lambda _i}{{\bf{u}}_i}}$.
Then, we have the following derivation:
\begin{align*}
    {\bf{G}}_{:,r}^T{ ({\bf{I}}_N - {\bf{L}}_\tau) }{{\bf{G}}_{:,r}} &= \sum\limits_{i = 1}^N {{\bf{G}}_{ir}^2}  - 2\sum\limits_{({v_i},{v_j}) \in E} {\frac{{{{\bf{G}}_{ir}}{{\bf{G}}_{jr}}}}{{\sqrt {{d_i} + \tau } \sqrt {{d_j} + \tau } }}} \\
    & = \sum\limits_{({v_i},{v_j}) \in E} {[{{(\frac{1}{{\sqrt {{d_i}} }}{{\bf{G}}_{ir}})}^2} - \frac{{2{{\bf{G}}_{ir}}{{\bf{G}}_{jr}}}}{{ \sqrt {{d_i} + \tau}} {\sqrt {{d_j} + \tau}} } + {{(\frac{1}{{ \sqrt {{d_j}} }}{{\bf{G}}_{jr}})}^2}]} \\
    & = \sum\limits_{({v_i},{v_j}) \in E({\hat S_r},V\backslash {\hat S_r})} {\frac{1}{{\mu ({{\hat S}_r})}}}  + \sum\limits_{({v_i},{v_j}) \in E({\hat S_r})} {\frac{2}{{\mu ({{\hat S}_r})}}(1 - \frac{{\sqrt {{d_i}{d_j}} }}{{\sqrt {{d_i} + \tau } \sqrt {{d_j} + \tau } }})} \\
    & \le \frac{{|E({\hat S_r},V\backslash {\hat S_r})|}}{{\mu ({\hat S_r})}} + \frac{{2|E({\hat S_r})|}}{{\mu ({\hat S_r})}}(1 - \frac{{{d_{\min }}}}{{{d_{\max }} + \tau}}) \\
    & = \phi ({{\hat S}_r}) + \frac{{\mu ({{\hat S}_r}) - |E({{\hat S}_r},V\backslash {{\hat S}_r})|}}{{\mu ({{\hat S}_r})}}(1 - \frac{{{d_{\min }}}}{{{d_{\max }} + \tau }})\\
    & = 1 - (1 - \phi ({\hat S}_r)) \frac{d_{\min}}{d_{\max}+ \tau}.
\end{align*}
Note that we also have
\begin{align*}
    {\bf{G}}_{:,r}^T{ ({\bf{I}}_{N} - {\bf{L}}_\tau) }{{\bf{G}}_{:,r}} & = {(\sum\limits_{i = 1}^N {{h_{ir}}{\lambda _i}{{\bf{u}}_i}} )^T}{({\bf{I}}_{N} - {\bf{L}}_\tau)}(\sum\limits_{i = 1}^N {{h_{ir}}{\lambda _i}{{\bf{u}}_i}} ) \\
    &= \sum\limits_{i = 1}^N {h_{ir}^2\lambda _i^2}  - \sum\limits_{i = 1}^N {h_{ir}^2\lambda _i^3} \\
    & = \sum\limits_{i = 1}^N {h_{ir}^2\lambda _i^2(1 - {\lambda _i})} \\
    & \ge \sum\limits_{i = K + 2}^N {h_{ir}^2\lambda _i^2(1 - {\lambda _i})} \\
    & \ge (1 - {\lambda _{K + 2}})\sum\limits_{i = K + 2}^N {h_{ir}^2\lambda _i^2}.
\end{align*}
By combining the aforementioned two inequalities, one can obtain
\begin{equation*}
    ||{{{\bf{\hat u}}}_r} - {{\bf{G}}_{:,r}}||_2^2 = \sum\limits_{i = K + 2}^N {h_{ir}^2\lambda _i^2}  \le \frac{1}{{1 - {\lambda _{K + 2}}}}[1 - (1 - \phi ({{\hat S}_r}))\frac{{{d_{\min }}}}{{{d_{\max }} + \tau }}],
\end{equation*}
for each $\hat S_r$ ($r \in \{ 1, \cdots, K\}$).
Let ${\bf {\hat F}} := [{\bf{\hat u}}_1, \cdots, {\bf{\hat u}}_K]$. For the whole graph $G$, we have
\begin{equation}\label{Eq:Aux-0}
    ||{\bf{\hat F}} - {\bf{G}}||_F^2 = \sum\limits_{r = 1}^K {||{{{\bf{\hat u}}}_r} - {{\bf{G}}_{:,r}}||_2^2}  \le \frac{K}{{1 - {\lambda _{K + 2}}}}[1 - (1 - {\bar \phi}_K (G))\frac{{{d_{\min }}}}{{{d_{\max }} + \tau }}] = K\Psi_{\rm{ISC}}.
\end{equation}
This inequality can be rewritten as
\begin{equation}\label{Eq:Aux-1}
    ||{\bf{\hat F}} - {\bf{G}}||_F^2 = ||{\bf{U \Lambda H}} - {\bf{G}}||_F^2 \le K\Psi_{\rm{ISC}},
\end{equation}
where ${\bf{H}} \in \mathbb{R}^{{(K+1)} \times K}$ is a matrix of rearranging weights $\{ h_{ir} \}$ in the linear combination.

\textbf{Claim.}
Let ${\bf{Z}} := {\bf{\hat F}} - {\bf{G}} = [{\bf{z}}_1, \cdots, {\bf{z}}_N]$. Then, ${\bf{Z}}^T {\bf{Z}} = {\bf{I}}_K - {\bf{H}}^T {\bf{\Lambda }}^T {\bf{\Lambda }} {\bf{H}}$.

\textbf{Proof}.
First, we have
\begin{equation}\label{Eq:Aux-2}
    {\bf{z}}_r^T{{\bf{z}}_s} = {({{{\bf{\hat u}}}_r} - {{\bf{G}}_{:,r}})^T}({{{\bf{\hat u}}}_s} - {{\bf{G}}_{:,s}}) = {\bf{\hat u}}_r^T{{{\bf{\hat u}}}_s} - {\bf{\hat u}}_r^T{{\bf{G}}_{:,s}} - {\bf{G}}_{:,r}^T{{{\bf{\hat u}}}_s} + {\bf{G}}_{:,r}^T{{\bf{G}}_{:,s}}.
\end{equation}
Consider each term in (\ref{Eq:Aux-2}). We further have
\begin{equation*}
    {\bf{\hat u}}_r^T{{{\bf{\hat u}}}_s} = {(\sum\limits_{t = 1}^{K + 1} {{h_{rt}}{\lambda _t}{{\bf{u}}_t}} )^T}(\sum\limits_{t = 1}^{K + 1} {{h_{st}}{\lambda _t}{{\bf{u}}_t}} ) = \sum\limits_{t = 1}^{K + 1} {{h_{rt}}{h_{st}}\lambda _t^2};
\end{equation*}
\begin{equation*}
    {\bf{\hat u}}_r^T{{\bf{G}}_{:,s}} = {(\sum\limits_{t = 1}^{K + 1} {{h_{rt}}{\lambda _t}{{\bf{u}}_t}} )^T}{(\sum\limits_{t = 1}^N {{h_{st}}{\lambda _t}{{\bf{u}}_t}} )^T} = \sum\limits_{t = 1}^{K + 1} {{h_{rt}}{h_{st}}\lambda _t^2};
\end{equation*}
\begin{equation*}
    {\bf{G}}_{:,r}^T{{{\bf{\hat u}}}_s} = {(\sum\limits_{t = 1}^N {{h_{rt}}{\lambda _t}{{\bf{u}}_t}} )^T}{(\sum\limits_{t = 1}^{K + 1} {{h_{st}}{\lambda _t}{{\bf{u}}_t}} )^T} = \sum\limits_{t = 1}^{K + 1} {{h_{rt}}{h_{st}}\lambda _t^2};
\end{equation*}
\begin{equation*}
{\bf{G}}_{:,r}^T{{\bf{G}}_{:,s}} = \left\{ {\begin{array}{*{20}{l}}
{1,r = s}\\
{0,r \ne s}
\end{array}} \right..
\end{equation*}
Therefore, one can rewrite (\ref{Eq:Aux-2}) as
\begin{equation*}
{\bf{z}}_r^T{{\bf{z}}_s} = \left\{ {\begin{array}{*{20}{l}}
{1 - \sum\nolimits_{t = 1}^{K + 1} {{h_{rt}}{h_{st}}\lambda _t^2} ,{\rm{~}}r = s}\\
{ - \sum\nolimits_{t = 1}^{K + 1} {{h_{rt}}{h_{st}}\lambda _t^2} ,{\rm{~~~~}} r \ne s}
\end{array}} \right.,
\end{equation*}
which corresponds to the following matrix form:
\begin{equation*}
    {{\bf{Z}}^T}{\bf{Z}} = {{\bf{I}}_K} - {{\bf{H}}^T}{{\bf{\Lambda }}^T}{\bf{\Lambda H}}.
\end{equation*}
This completes the proof of \textbf{Claim}.

Consider the singular value decomposition of ${\bf{H}} \in \mathbb{R}^{(K+1) \times K}$ denoted as ${\bf{H}} = {\bf{X}} {\bf{\Sigma}} {\bf{Y}}^T$, where ${\bf{X}} \in \mathbb{R}^{(K+1) \times (K+1)}$ and ${\bf{Y}} \in \mathbb{R}^{K \times K}$ are orthogonal matrices; ${\bf{\Sigma}}$ is a $(K+1) \times K$ matrix, with only the $(i, i)$-th entries as non-zero values $\{ \sigma_i \}$. 
One can rewrite (\ref{Eq:Aux-1}) as
\begin{align*}
    K\Psi_{\rm{ISC}}  & \ge ||{\bf{\hat F}} - {\bf{G}}||_F^2 = ||{\bf{Z}}||_F^2 \ge ||{{\bf{Z}}^T}{\bf{Z}}|{|_F} \\
    & = ||{{\bf{I}}_K} - {{\bf{H}}^T}{{\bf{\Lambda }}^2}{\bf{H}}|{|_F} \\
    & = ||{{\bf{I}}_K} - {\bf{Y\Sigma }}{{\bf{X}}^T}{{\bf{\Lambda }}^2}{\bf{X\Sigma }}{{\bf{Y}}^T}|{|_F} \\
    & = || {\bf{Y}} ({{\bf{I}}_K} - {\bf{\Sigma }}{{\bf{X}}^T}{{\bf{\Lambda }}^2}{\bf{X\Sigma }}) {{\bf{Y}}^T} |{|_F} \\
    & = ||{{\bf{I}}_K} - {\bf{\Sigma }}{{\bf{X}}^T}{{\bf{\Lambda }}^2}{\bf{X\Sigma }}|{|_F}.
\end{align*}
Note that ${({\bf{\Lambda X\Sigma }})_{ij}} = {\lambda _i}{\sigma _j}{{\bf{X}}_{ij}}$, which leads to
\begin{equation*}
    {({\bf{\Sigma }}{{\bf{X}}^T}{{\bf{\Lambda }}^2}{\bf{X\Sigma }})_{ij}} = {\sigma _i}{\sigma _j}\sum\nolimits_{l = 1}^{K + 1} {\lambda _l^2{{\bf{X}}_{li}}{{\bf{X}}_{lj}}}  \le {\sigma _i}{\sigma _j}\lambda _1^2\sum\nolimits_{l = 1}^{K + 1} {{{\bf{X}}_{li}}{{\bf{X}}_{lj}}}  = {\sigma _i}{\sigma _j}\lambda _1^2 = \left\{ {\begin{array}{*{20}{c}}
{\lambda _1^2\sigma _i^2,i = j}\\
{0,{\rm{~~~~~~}}x \ne j}
\end{array}} \right..
\end{equation*}
We further have
\begin{equation}\label{Eq:Aux__0}
    K\Psi_{\rm{ISC}} \ge ||{{\bf{I}}_K} - {\bf{\Sigma }}{{\bf{X}}^T}{{\bf{\Lambda }}^2}{\bf{X\Sigma }}||_F \ge {[\sum\limits_{r = 1}^K {(1 - \lambda _1^2\sigma _r^2)^2} ]^{1/2}} \ge {[\sum\limits_{r = 1}^K {(1 - \sigma _r^2)^2} ]^{1/2}} = || {\bf{I}}_{(K+1) \times K} - {\bf{\Sigma}}^2 ||_F,
\end{equation}
where ${{\bf{I}}_{(K + 1) \times K}}: = {[{{\bf{I}}_K},{{\bf{0}}_{K \times 1}}]^T}$.
Let ${\bf{O}} := {\bf{X}} {{\bf{I}}_{(K + 1) \times K}} {\bf{Y}}^T$, which is an orthogonal matrix. One can obtain
\begin{align*}
    ||{\bf{FO}} - {\bf{G}}|{|_F} & = ||{\bf{U\Lambda O}} - {\bf{G}}|{|_F} \\
    & = ||{\bf{U\Lambda O}} - {\bf{\hat F}} + {\bf{\hat F}} - {\bf{G}}|{|_F} \\
    & \le ||{\bf{U\Lambda O}} - {\bf{\hat F}}|{|_F} + ||{\bf{\hat F}} - {\bf{G}}|{|_F} \\
    & = ||{\bf{U}}({\bf{\Lambda O}} - {\bf{\Lambda H}})|{|_F} + ||{\bf{\hat F}} - {\bf{G}}|{|_F} \\
    & \le ||{\bf{\Lambda O}} - {\bf{\Lambda H}}|{|_F} + ||{\bf{\hat F}} - {\bf{G}}|{|_F}
\end{align*}
For the first term, one can derive the following inequality:
\begin{align*}
    ||{\bf{\Lambda O}} - {\bf{\Lambda H}}|{|_F}  & = || {\bf{\Lambda }} {\bf{X}}({{\bf{I}}_{(K + 1) \times K}} - {\bf{\Sigma }}){{\bf{Y}}^T}|{|_F} \\
    & = ||{\bf{\Lambda }} ({{\bf{I}}_{(K + 1) \times K}} - {\bf{\Sigma }})|{|_F} \\
    & = \sqrt { {\lambda_1}^2 {{(1 - {\sigma _1})}^2} +  \cdots  + {\lambda_K}^2 {{(1 - {\sigma _K})}^2}} \\
    & \le \sqrt { {\lambda_1}^2 [{{(1 - {\sigma _1})}^2} +  \cdots  + {{(1 - {\sigma _K})}^2}] } \\
    & \le {\lambda_1} \sqrt {{{(1 - {\sigma _1})}^2}{{(1 + {\sigma _1})}^2} +  \cdots  + {{(1 - {\sigma _K})}^2}{{(1 + {\sigma _K})}^2}} \\
    & = {\lambda_1} \sqrt {{{(1 - \sigma _1^2)}^2} +  \cdots  + {{(1 - \sigma _K^2)}^2}} \\
    & = {\lambda_1} ||{\bf{I}} - {{\bf{\Sigma }}^2}|{|_F} \le {\lambda _1}K\Psi_{\rm{ISC}} {\rm{~(By~applying~(\ref{Eq:Aux__0}))}}.
\end{align*}
By combining (\ref{Eq:Aux-0}), we finally have
\begin{equation*}
    ||{\bf{FO}} - {\bf{G}}|{|_F} \le ||{\bf{\Lambda O}} - {\bf{\Lambda H}}|{|_F} + ||{\bf{\hat F}} - {\bf{G}}|{|_F} \le {\lambda _1}K\Psi_{\rm{ISC}}  + \sqrt {K\Psi_{\rm{ISC}}}.
\end{equation*}
If $K\Psi_{\rm{ISC}} \le 1$, then
\begin{equation*}
    ||{\bf{F}} {\bf{O}} - {\bf{G}}||_F \le (1 + \lambda_1) \sqrt{K \Psi_{\rm{ISC}}}.
\end{equation*}
If $K\Psi_{\rm{ISC}} > 1$, we also have
\begin{equation*}
    ||{\bf{F}} {\bf{O}} - {\bf{G}}||_F \le (1 + \lambda_1) K \Psi_{\rm{ISC}}.
\end{equation*}
This completes the proof of \textbf{Theorem~\ref{Th:Struc}}.

\section{Proof of Lemma~\ref{le1}}\label{App:Lemma1}
Since ${\bf{o}}_r$ and ${\bf{o}}_t$ ($\forall r \ne t$) are orthogonal, we have
\begin{equation*}
    ||{{\bf{o}}_r} - {{\bf{o}}_t}||_2^2 = {({{\bf{o}}_r} - {{\bf{o}}_t})^T}({{\bf{o}}_r} - {{\bf{o}}_t}) = ||{{\bf{o}}_r}||_2^2 + ||{{\bf{o}}_r}||_2^2 - {\bf{o}}_t^T{{\bf{o}}_r} - {\bf{o}}_r^T{{\bf{o}}_t} = 2.
\end{equation*}
According to Lemma 1 in \cite{Mizutani}, one can obtain
\begin{equation}\label{Eq:Aux-3}
    {\left\| {\frac{{\bf{a}}}{{||{\bf{a}}|{|_2}}} - {\bf{b}}} \right\|_2} \le 2||{\bf{a}} - {\bf{b}}|{|_2},
\end{equation}
for vectors ${\bf{a}}, {\bf{b}} \in \mathbb{R}^{K+1}$.
By (\ref{Eq:Dist}), (\ref{Eq:Struc}), and \textbf{Theorem~\ref{Th:Struc}}, we have the following derivation:
\begin{align*}
    g({{\hat S}_1}, \cdots ,{{\hat S}_K};{o_1}, \cdots ,{o_K}) &= \sum\limits_{r = 1}^K {\sum\limits_{{v_i} \in {{\hat S}_r}} {{d_i}\left\| {{{{\bf{\tilde F}}}_{i,:}} - {{\bf{o}}_r}} \right\|_2^2} }\\
    &= \sum\limits_{r = 1}^K {\sum\limits_{{v_i} \in {{\hat S}_r}} {{d_i}\left\| {\frac{{{{\bf{F}}_{i,:}}}}{{||{{\bf{F}}_{i,:}}||_2}} - {{\bf{o}}_r}} \right\|_2^2} } \\
    &= \sum\limits_{r = 1}^K {\sum\limits_{{v_i} \in {{\hat S}_r}} {{d_i}\left\| {\frac{{{{\bf{F}}_{i,:}}\sqrt {\mu ({{\hat S}_r})/{d_i}} }}{{||{{\bf{F}}_{i,:}}\sqrt {\mu ({{\hat S}_r})/{d_i}} ||_2}} - {{\bf{o}}_r}} \right\|_2^2} }\\
    &\le 4\sum\limits_{r = 1}^K {\sum\limits_{{v_i} \in {{\hat S}_r}} {{d_i}\left\| {{{\bf{F}}_{i,:}}\sqrt {\frac{{\mu ({{\hat S}_r})}}{{{d_i}}}}  - {{\bf{o}}_r}} \right\|_2^2} } {\rm{(By~applying~(\ref{Eq:Aux-3}))}}\\
    &=4\sum\limits_{r = 1}^K {\sum\limits_{{v_i} \in {{\hat S}_r}} {\mu ({{\hat S}_r})} } \left\| {{{\bf{F}}_{i,:}} - \sqrt {\frac{{{d_i}}}{{\mu ({{\hat S}_r})}}} {{\bf{o}}_r}} \right\|_2^2\\
    &\le 4\sum\limits_{r = 1}^K {\sum\limits_{{v_i} \in {{\hat S}_r}} {{\mu _{\max }}} } \left\| {{{\bf{F}}_{i,:}} - \sqrt {\frac{{{d_i}}}{{\mu ({{\hat S}_r})}}} {{\bf{o}}_r}} \right\|_2^2\\
    &\le 4{(1 + {\lambda _1})^2}{\mu _{\max }}K\Psi_{\rm{ISC}}.
\end{align*}
This completes the proof of \textbf{Lemma~\ref{le1}}.

\section{Proof of Lemma~\ref{le2}}\label{App:Lemma2}

Based on the definition of $H_{\pi, r}$, we have
\begin{equation}\label{Eq:Aux-4}
    \sum\limits_{r = 1}^K {\mu ({H_{\pi ,r}}\Delta {{\hat S}_r})}  = 2\sum\limits_{r = 1}^K {\sum\limits_{1 \le t \le K,t \ne \pi (r)} {\mu ({C_r} \cap {{\hat S}_t})} }.
\end{equation}
For any vectors ${\bf{a}}$, ${\bf{b}}$, ${\bf{c}} \in \mathbb{R}^{K+1}$, \cite{peng} have proved that
\begin{equation}\label{Eq:Aux_}
    ||{\bf{a}}||_2^2 \ge \frac{1}{2}||{\bf{a}} - {\bf{b}}||_2^2 - ||{\bf{b}}||_2^2,
\end{equation}
\begin{equation*}
    ||{\bf{a}} - {\bf{c}}|{|_2} \ge ||{\bf{b}} - {\bf{c}}|{|_2} \Rightarrow ||{\bf{a}} - {\bf{c}}|{|_2} \ge ||{\bf{a}} - {\bf{b}}|{|_2}/2.
\end{equation*}
Since we always have $||{{\bf{o}}_t} - {{\bf{c}}_r}||_2^2 \ge ||{{\bf{o}}_{\pi (r)}} - {{\bf{c}}_r}||_2^2$ following the definition of $\pi(r)$, we have
\begin{equation}\label{Eq:Aux-5}
    ||{{\bf{o}}_t} - {{\bf{c}}_r}||_2^2 \ge \frac{1}{4}||{{\bf{o}}_t} - {{\bf{o}}_{\pi (r)}}||_2^2 = \frac{1}{2} - \frac{1}{2}{\bf{o}}_t^T{{\bf{o}}_{\pi (r)}} = \left\{ {\begin{array}{*{20}{c}}
{0,{\rm{~~~~}}t = \pi (r)}\\
{1/2,{\rm{~}}t \ne \pi (r)}
\end{array}} \right..
\end{equation}

\begin{align*}
    & {\mathop{\rm COST}\nolimits} ({C_1}, \cdots ,{C_K}) = \sum\limits_{r = 1}^K {\sum\limits_{{v_i} \in {C_r}} {{d_i}||{{{\bf{\tilde F}}}_{i,:}} - {{\bf{c}}_r}||_2^2} }\\
    & \ge \sum\limits_{r = 1}^K {\sum\limits_{1 \le t \le K,t \ne \pi (r)} {\sum\limits_{{v_i} \in {C_r} \cap {{\hat S}_t}} {{d_i}||{{{\bf{\tilde F}}}_{i,:}} - {{\bf{c}}_r}||_2^2} } } \\
    & \ge \sum\limits_{r = 1}^K {\sum\limits_{1 \le t \le K,t \ne \pi (r)} {\sum\limits_{{v_i} \in {C_r} \cap {{\hat S}_t}} {{d_i}[\frac{{||{{\bf{o}}_t} - {{\bf{c}}_r}||_2^2}}{2} - ||{{{\bf{\tilde F}}}_{i,:}} - {{\bf{o}}_t}||_2^2]} } } {\rm{~(By~applying~(\ref{Eq:Aux_}))}} \\
    & \ge \sum\limits_{r = 1}^K {\sum\limits_{1 \le t \le K,t \ne \pi (r)} {\sum\limits_{{v_i} \in {C_r} \cap {{\hat S}_t}} {\frac{{{d_i}}}{4}} } }  - \sum\limits_{r = 1}^K {\sum\limits_{1 \le t \le K,t \ne \pi (r)} {\sum\limits_{{v_i} \in {C_r} \cap {{\hat S}_t}} {{d_i}||{{{\bf{\tilde F}}}_{i,:}} - {{\bf{o}}_t}||_2^2} } } {\rm{~(By~applying~(\ref{Eq:Aux-5}))}} \\
    & = \frac{1}{4}\sum\limits_{r = 1}^K {\sum\limits_{1 \le t \le K,t \ne \pi (r)} {\mu ({C_r} \cap {{\hat S}_t})} }  - \sum\limits_{r = 1}^K {\sum\limits_{1 \le t \le K,t \ne \sigma (r)} {\sum\limits_{{v_i} \in {C_r} \cap {{\hat S}_t}} {{d_i}||{{{\bf{\tilde F}}}_{i,:}} - {{\bf{o}}_t}||_2^2} } } \\
    & \ge \frac{1}{4}\sum\limits_{r = 1}^K {\sum\limits_{1 \le t \le K,t \ne \pi (r)} {\mu ({C_r} \cap {{\hat S}_t})} }  - \sum\limits_{r = 1}^K {\sum\limits_{{v_i} \in {{\hat S}_r}} {{d_i}||{{{\bf{\tilde F}}}_{i,:}} - {{\bf{o}}_r}||_2^2} } \\
    & \ge \frac{1}{4}\sum\limits_{r = 1}^K {\sum\limits_{1 \le t \le K,t \ne \pi (r)} {\mu ({C_r} \cap {{\hat S}_t})} }  - 4{(1 + {\lambda _1})^2}{\mu _{\max }}K{\Psi _{{\rm{ISC}}}} {\rm{~(By~Lemma~\ref{le1}'s~proof)}}\\
    & \ge \frac{1}{8}\sum\limits_{r = 1}^K {\mu ({H_{\pi ,r}}\Delta {{\hat S}_r})}  - 4{(1 + {\lambda _1})^2}{\mu _{\max }}K{\Psi _{{\rm{ISC}}}} {\rm{~(By~applying~(\ref{Eq:Aux-4}))}}.
\end{align*}
We then have
\begin{align*}
    \sum\limits_{r = 1}^K {\mu ({H_{\pi ,r}}\Delta {{\hat S}_r})}  & \le 8{\mathop{\rm COST}\nolimits} ({C_1}, \cdots ,{C_K}) + 32{(1 + {\lambda _1})^2}{\mu _{\max }}K{\Psi _{{\rm{ISC}}}} \\
    & \le 32{(1 + {\lambda _1})^2}\alpha {\mu _{\max }}K{\Psi _{{\rm{ISC}}}} + 32{(1 + {\lambda _1})^2}{\mu _{\max }}K{\Psi _{{\rm{ISC}}}} {\rm{~(By~Theorem~\ref{th2-1})}}\\
    & = 32(1 + \alpha ){(1 + {\lambda _1})^2}{\mu _{\max }}K{\Psi _{{\rm{ISC}}}}.
\end{align*}
This completes the proof of \textbf{Lemma~\ref{le2}}.

\section{Proof of Theorem~\ref{Th:Main}}\label{App:Th-Main}
Following the same definitions of $\pi (r)$ and $H_{\pi, r}$ in \textbf{Lemma~\ref{le2}}, $\pi$ may be (\romannumeral1) a permutation or (\romannumeral2) not be a permutation, where the latter the case results in higher upper bound for the mis-clustered volume $\sum\nolimits_{r = 1}^K {\mu ({C_r}\Delta {{\hat S}_r})}$ compared with the former one. Hence, the final upper bound of \textbf{Theorem~\ref{Th:Main}} is dominated by the case when $\pi$ is not a permutation regarding $({\bf{o}}_1, \cdots, {\bf{o}}_K)$.

When $\pi$ is not a permutation, there exist (\romannumeral1) an index $r \in [K]$ such that $H_{\pi, r} = \emptyset$ and (\romannumeral2) different indices $a, b \in [K]$ such that $\pi(a) = \pi(b) = t$ and $t \ne r$. We adopt a strategy similar to the proof of Theorem 2 in \cite{macgregor2022tighter} to construct a permutation from $\pi$ by iteratively taking the following steps. In each iteration, we introduce a function $\pi' : [K] \to [K]$ given $\pi$, where we
\begin{itemize}
    \item set $\pi' (c) = r$ if $c = a$;
    \item set $\pi '(c) = \pi (c)$ for any other indices $c \in [K] \backslash \{ a \}$.
\end{itemize}
In particular, the construction of $\pi'$ (i.e., one iteration) can reduce the number of $H_{\pi, r} = \emptyset$ by one. By iteratively taking the aforementioned steps, we can finally remove all these empty sets and get a permutation.
Let $\bar \mu_{\pi} := \sum\nolimits_{r = 1}^K {\mu ({H_{\pi, r}}\Delta {{\hat S}_r})}$. By $\textbf{Lemma~\ref{le2}}$, we have
\begin{equation}\label{Eq:Aux-6}
    \bar \mu_{\pi} \le 32 (1+\alpha)(1+\lambda_1)^2\mu_{\max}K \Psi_{\rm{ISC}}.
\end{equation}
Let $\pi^*$ be the final permutation we obtained. Our goal is to derive the upper bound of $\bar \mu_{\pi^*} = \sum\nolimits_{r = 1}^K {\mu ({C_r}\Delta {{\hat S}_r})}$.
Before that, we first consider the upper bound of the difference between $\bar \mu_{\pi'}$ and $\bar \mu_{\pi}$ in one iteration, i.e.,
\begin{equation*}
    \bar \mu_{\pi'} - \bar \mu_{\pi} = [\underbrace {\mu ({H_{\pi ',r}}\Delta {{\hat S}_r}) - \mu ({H_{\pi ,r}}\Delta {{\hat S}_r})}_{: = {p_1}}] + [\underbrace {\mu ({H_{\pi ',t}}\Delta {{\hat S}_t}) - \mu ({H_{\pi ,t}}\Delta {{\hat S}_t})}_{{:= p_2}}].
\end{equation*}
Fig.~\ref{Fig:Proof-Exp} illustrates the relations among $H_{\pi, r} \Delta \hat S_r$, $H_{\pi, t} \Delta \hat S_t$, $H_{\pi',r} \Delta \hat S_r$, and $H_{\pi',t} \Delta \hat S_t$.
According to the signs of $p_1$ and $p_2$, there exist the following $4$ cases.

\begin{figure}
  \centering
  \includegraphics[width=0.7\linewidth, trim=18 22 20 18,clip]{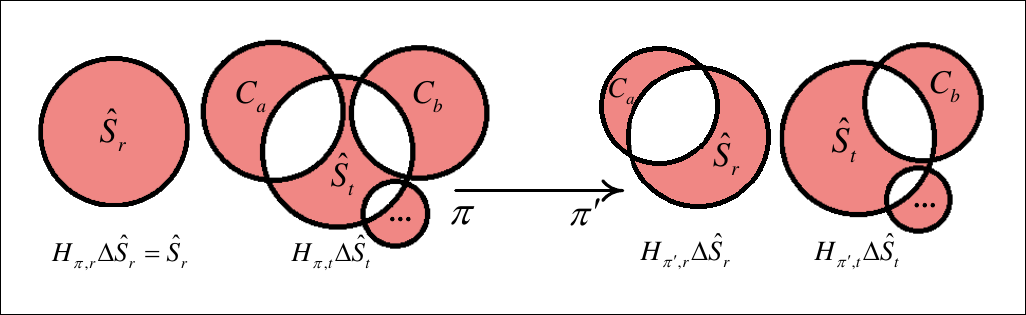}
  \caption{Illustrations of $H_{\pi, r} \Delta \hat S_r$, $H_{\pi, t} \Delta \hat S_t$, $H_{\pi',r} \Delta \hat S_r$, and $H_{\pi',t} \Delta \hat S_t$ denoted by red areas.
  }\label{Fig:Proof-Exp}
\end{figure}

(\textbf{Case 1}) $p_1 < 0$ and $p_2 < 0$. We can directly derive $\bar \mu_{\pi'} - \bar \mu_{\pi} \le 0$, which has no contribution to the upper bound.

(\textbf{Case 2}) $p_1 > 0$ and $p_2 < 0$. We have
\begin{align*}
    {{\bar \mu }_{\pi '}} - {{\bar \mu }_\pi } & \le {p_1} + |{p_2}| = {p_1} - {p_2} \\
    & = \mu ({C_a}\backslash {{\hat S}_r}) - \mu ({C_a} \cap {{\hat S}_r}) + \mu ({C_a}\backslash {{\hat S}_t}) - \mu ({C_a} \cap {{\hat S}_t}) \\
    & \le \mu ({C_a}\backslash {{\hat S}_r}) + \mu ({C_a}\backslash {{\hat S}_t}) \\
    & \le 2\max \{ \mu ({C_a}\backslash {{\hat S}_r}),\mu ({C_a}\backslash {{\hat S}_t})\}.
\end{align*}
Without loss of generality, we assume that $\max \{ \mu ({C_a}\backslash {{\hat S}_r}),\mu ({C_a}\backslash {{\hat S}_t})\}  = \mu ({C_a}\backslash {{\hat S}_t})$. Since each iteration moves a selected set $C_a$ at most once, the total contribution of case 2 to the upper bound is at most
\begin{equation*}
    \sum\limits_{t = 1}^K {\sum\limits_{{C_a} \in {H_{\pi ,t}}} {2\mu ({C_a}\backslash {{\hat S}_t})} }  \le 2\sum\limits_{t = 1}^K {\mu ({H_{\pi ,t}}\Delta {{\hat S}_t})}  \le 64(1 + \alpha ){(1 + {\lambda _1})^2}{\mu _{\max }}K{\Psi _{{\rm{ISC}}}},
\end{equation*}
by applying \textbf{Lemma~\ref{le2}}.

(\textbf{Case 3}) $p_1 > 0$ and $p_2 > 0$. Following the assumption $\max \{ \mu ({C_a}\backslash {{\hat S}_r}),\mu ({C_a}\backslash {{\hat S}_t})\}  = \mu ({C_a}\backslash {{\hat S}_t})$ in case 2, we have
\begin{equation}\label{Eq:Aux-7}
    \mu ({C_a}\backslash {{\hat S}_t}) \ge \mu ({C_a}\backslash {{\hat S}_r}) \Rightarrow \mu ({C_a}) - \mu ({C_a}\backslash {{\hat S}_t}) \le \mu ({C_a}) - \mu ({C_a}\backslash {{\hat S}_r}) \Rightarrow \mu ({C_a} \cap {{\hat S}_r}) \ge \mu ({C_a} \cap {{\hat S}_t}).
\end{equation}
Since $p_2 > 0$, $\mu ({C_a} \cap {{\hat S}_t}) - \mu ({C_a}\backslash {{\hat S}_t}) > 0 \Rightarrow \mu ({C_a} \cap {{\hat S}_t}) > \mu ({C_a}\backslash {{\hat S}_t})$. Combining it with (\ref{Eq:Aux-7}), we then have
\begin{equation}\label{Eq:Aux__}
    \mu ({C_a} \cap {{\hat S}_t}) + \mu ({C_a} \cap {{\hat S}_r}) \ge \mu ({C_a}\backslash {{\hat S}_t}) + \mu ({C_a} \cap {{\hat S}_t}) \Rightarrow 2\mu ({C_a} \cap {{\hat S}_r}) \ge \mu ({C_a}\backslash {{\hat S}_t}) \ge \mu ({C_a}\backslash {{\hat S}_r}).
\end{equation}
For the difference between $\bar \mu_{\pi'}$ and $\bar \mu_{\pi}$, we have the following derivations:
\begin{align*}
    {{\bar \mu }_{\pi '}} - {{\bar \mu }_\pi } & = {p_1} + {p_2} \\
    & = \mu ({C_a}\backslash {{\hat S}_r}) - \mu ({C_a} \cap {{\hat S}_r}) + \mu ({C_a} \cap {{\hat S}_t}) - \mu ({C_a}\backslash {{\hat S}_t}) \\
    & \le 2\mu ({C_a} \cap {{\hat S}_r}) - \mu ({C_a} \cap {{\hat S}_r}) + \mu ({C_a} \cap {{\hat S}_t}) {\rm{~(By~applying~(\ref{Eq:Aux__}))}}\\
    & = \mu ({C_a} \cap {{\hat S}_r}) + \mu ({C_a} \cap {{\hat S}_t})\\
    & \le 2\mu ({C_a} \cap {{\hat S}_r}) {\rm{~(By~applying~(\ref{Eq:Aux-7}))}}\\
    & \le 2\mu (\hat S_r).
\end{align*}

(\textbf{Case 4}) $p_1 < 0$ and $p_2 > 0$. We can also derive (\ref{Eq:Aux-7}) for this case based on the assumption $\max \{ \mu ({C_a}\backslash {{\hat S}_r}),\mu ({C_a}\backslash {{\hat S}_t})\}  = \mu ({C_a}\backslash {{\hat S}_t})$ in case 2. Then, we can obtain
\begin{align*}
    {{\bar \mu }_{\pi '}} - {{\bar \mu }_\pi } & \le |{p_1}| + {p_2} = - {p_1} + {p_2} \\
    & = \mu ({C_a} \cap {{\hat S}_r}) - \mu ({C_a}\backslash {{\hat S}_r}) + \mu ({C_a} \cap {{\hat S}_t}) - \mu ({C_a}\backslash {{\hat S}_t}) \\
    & \le \mu ({C_a} \cap {{\hat S}_r}) + \mu ({C_a} \cap {{\hat S}_t}) \le 2 \mu ({C_a} \cap {{\hat S}_r}) {\rm{~(By~applying~(\ref{Eq:Aux-7}))}} \\
    & \le 2 \mu (\hat S_r).
\end{align*}
Note that for each $r \in \Omega := \{ r: H_{\pi, r} = \emptyset\}$, we have
\begin{equation}\label{Eq:Aux-8}
    \mu ({H_{\pi ,r}}\Delta {{\hat S}_r}) = \mu ({{\hat S}_r}).
\end{equation}
Since each iteration can reduce the number of $H_{\pi,r} = \emptyset$ by one, the number of iterations to derive a feasible permutation is bounded by $|\Omega|$.
Then, we can derive the following total contribution of case 3 and case 4 to the final upper bound:
\begin{equation*}
    \sum\limits_{r \in \Omega } {2\mu ({{\hat S}_r})}  = 2\sum\limits_{r \in \Omega } {\mu ({H_{\pi ,r}}\Delta {{\hat S}_r})}  \le 2\sum\limits_{r = 1}^K {\mu ({H_{\pi ,r}}\Delta {{\hat S}_r})}  \le 64(1 + \alpha ){(1 + {\lambda _1})^2}{\mu _{\max }}K{\Psi _{{\rm{ISC}}}},
\end{equation*}
by applying \textbf{Lemma~\ref{le2}}, with $\mu_{\max} := \max \{ \mu (\hat S)\}$. Considering the four cases, we have
\begin{align*}
    {\bar \mu _{{\pi ^*}}} & \le {\bar \mu _\pi } + 128(1 + \alpha ){(1 + {\lambda _1})^2}{\mu _{\max }}K{\Psi _{{\rm{ISC}}}} \\
    & \le 160(1 + \alpha ){(1 + {\lambda _1})^2}{\mu _{\max }}K{\Psi _{{\rm{ISC}}}} {\rm{~(By~applying (\ref{Eq:Aux-6}))}}.
\end{align*}
Namely, we have proved that
\begin{equation*}
    \sum\nolimits_{r = 1}^K {\mu ({C_r}\Delta {{\hat S}_r})}  \le 160(1 + \alpha ){(1 + {\lambda _1})^2}{\mu _{\max }}K{\Psi _{{\rm{ISC}}}}.
\end{equation*}
For the mis-clustered volume $\sum\nolimits_{r = 1}^K {\mu ({C_r}\Delta {{\hat S}_r})}$, one can derive that ${d_{\min }}\sum\nolimits_{r = 1}^K {|{C_r}\Delta {{\hat S}_r}|}  \le \sum\nolimits_{r = 1}^K {\mu ({C_r}\Delta {{\hat S}_r})}$.
Let $\mathcal{M}$ denote the set of mis-clustered nodes w.r.t. the optimal partition. We finally have
\begin{equation*}
    |\mathcal{M}| \le \sum\nolimits_{r = 1}^K {|{C_r}\Delta {{\hat S}_r}|}  \le \sum\nolimits_{r = 1}^K {\mu ({C_r}\Delta {{\hat S}_r})}  \le 160(1 + \alpha ){(1 + {\lambda _1})^2}{\tilde \mu}K{\Psi _{{\rm{ISC}}}},
\end{equation*}
with $\tilde \mu := \mu_{\max} / d_{\min}$.
This completes the proof of \textbf{Theorem~\ref{Th:Main}}.

\section{Proof of Proposition~\ref{Th:ASCENT}}\label{App:Th-ASCENT}
Similar to the derivation of the first inequality in Appendix~\ref{App:Th-Struc}, we have the following derivations for ASCENT:
\begin{align*}
    {\bf{G}}_{:,r}^T({{\bf{I}}_N} - {{\bf{L}}_\tau}){{\bf{G}}_{:,r}} &= \sum\limits_{i = 1}^N {{\bf{G}}_{ir}^2}  - 2\sum\limits_{({v_i},{v_j}) \in E} {\frac{{{{\bf{G}}_{ir}}{{\bf{G}}_{jr}}}}{{\sqrt {{d_i} + {\tau _i}} \sqrt {{d_j} + {\tau _j}} }}}, \\
    & = \sum\limits_{({v_i},{v_j}) \in E} {[{{(\frac{1}{{\sqrt {{d_i}} }}{{\bf{G}}_{ir}})}^2} - \frac{{2{{\bf{G}}_{ir}}{{\bf{G}}_{jr}}}}{{\sqrt {{d_i} + {\tau _i}} \sqrt {{d_i} + {\tau _i}} }} + [{{(\frac{1}{{\sqrt {{d_j}} }}{{\bf{G}}_{jr}})}^2}]} \\
    & = \sum\limits_{({v_i},{v_j}) \in E({S_r},V\backslash {{\hat S}_r})} {\frac{1}{{\mu ({{\hat S}_r})}}}  + \sum\limits_{({v_i},{v_j}) \in E({{\hat S}_r})} {\frac{2}{{\mu ({{\hat S}_r})}}} (1 - \frac{{\sqrt {{d_i}{d_j}} }}{{\sqrt {{d_i} + {\tau _i}} \sqrt {{d_j} + {\tau _j}} }}).
\end{align*}
Let ${\hat \tau _r}: = \max \{ {\tau _i}|{v_i} \in {\hat S_r}\}$, i.e., the maximum corrections among nodes in cluster $\hat S_r$. Then, we have the following derivations:
\begin{align*}
    {\bf{G}}_{:,r}^T({{\bf{I}}_N} - {{\bf{L}}_\tau }){{\bf{G}}_{:,r}} & \le \sum\limits_{({v_i},{v_j}) \in E({S_r},V\backslash {{\hat S}_r})} {\frac{1}{{\mu ({{\hat S}_r})}}}  + \sum\limits_{({v_i},{v_j}) \in E({{\hat S}_r})} {\frac{2}{{\mu ({{\hat S}_r})}}} (1 - \frac{{\sqrt {{d_i}{d_j}} }}{{\sqrt {{d_i} + {{\hat \tau }_r}} \sqrt {{d_j} + {{\hat \tau }_r}} }}) \\
    & \le \frac{{|E({{\hat S}_r},V\backslash {{\hat S}_r})|}}{{\mu ({{\hat S}_r})}} + \frac{{2|E({{\hat S}_r})|}}{{\mu ({{\hat S}_r})}}(1 - \frac{{{d_{\min }}}}{{{d_{\max }} + {{\hat \tau }_r}}}) \\
    & = \phi ({{\hat S}_r}) + \frac{{\mu ({{\hat S}_r}) - |E({{\hat S}_r},V\backslash {{\hat S}_r})|}}{{\mu ({{\hat S}_r})}}(1 - \frac{{{d_{\min }}}}{{{d_{\max }} + {{\hat \tau }_r}}}) \\
    & = \phi ({{\hat S}_r}) + (1 - \phi ({{\hat S}_r}))(1 - \frac{{{d_{\min }}}}{{{d_{\max }} + {{\hat \tau }_r}}}) \\
    & = 1 - (1 - \phi ({{\hat S}_r}))\frac{{{d_{\min }}}}{{{d_{\max }} + {{\hat \tau }_r}}}.
\end{align*}
Following the same definitions of $\{ h_{ir} \}$, $\{ {\bf{\hat u}}_r \}$, and ${\bf{\hat F}}$ in Appendix~\ref{App:Th-Struc}, we further have
\begin{equation*}
    ||{{\bf{\hat u}}_r} - {{\bf{G}}_{:,r}}||_2^2 = \sum\limits_{i = K + 2}^N {h_{ir}^2\lambda _i^2 \le \frac{1}{{1 - {\lambda _{K + 2}}}}[1 - (1 - \phi ({{\hat S}_r}))\frac{{{d_{\min }}}}{{{d_{\max }} + {{\hat \tau }_r}}}]},
\end{equation*}
for each cluster $\hat S_r$. For the whole graph, we have
\begin{align*}
    ||{\bf{\hat F}} - {\bf{G}}||_F^2 & = \sum\limits_{r = 1}^K {||{{{\bf{\hat u}}}_r} - {{\bf{G}}_{:,r}}||_2^2} \\
    & \le \frac{K}{{1 - {\lambda _{K + 2}}}} - \frac{1}{{1 - {\lambda _{K + 2}}}}\sum\limits_{r = 1}^K {(\frac{{{d_{\min }}}}{{{d_{\max }} + {{\hat \tau }_r}}} - \frac{{{d_{\min }}}}{{{d_{\max }} + {{\hat \tau }_r}}}\phi ({{\hat S}_r}))} \\
    & = K\left[ {\frac{1}{{1 - {\lambda _{K + 2}}}} - \frac{1}{{K(1 - {\lambda _{K + 2}})}}\sum\limits_{r = 1}^K {\frac{{{d_{\min }}(1 - \phi ({{\hat S}_r}))}}{{{d_{\max }} + {{\hat \tau }_r}}}} } \right] \\
    & = K \Psi_{\rm{AST}}.
\end{align*}
By following the same strategy of the proof of \textbf{Theorem~\ref{Th:Struc}}, \textbf{Theorem~\ref{th2-1}}, \textbf{Lemma~\ref{le2}}, and \textbf{Theorem~\ref{Th:Main}}, we can complete the proof of \textbf{Proposition~\ref{Th:ASCENT}}.

\printcredits

\bibliographystyle{elsarticle-num}



\end{document}